\pdfoutput=1
\documentclass[atmp,pdfpagelabels=true]{ipart_v1}

\Vol{11}
\Issue{1}
\Year{2023}
\firstpage{1}

\usepackage[english]{babel}

\usepackage{amsmath,amsfonts,amssymb,amsthm,amscd}

\usepackage{bbm}
\usepackage{bm}
\usepackage{mathrsfs}
\usepackage{fullpage}
\usepackage{multirow,bigdelim} 
\usepackage{makecell} 
\usepackage{tikz-cd}

\newcommand{\be}[0]{\begin{equation}}
\newcommand{\ee}[0]{\end{equation}}
\DeclareMathAlphabet{\mathpzc}{OT1}{pzc}{m}{it}
\DeclareMathAlphabet{\mathpzc}{OT1}{pzc}{m}{it}
\newcommand{\A}{\ensuremath{\mathcal{A}}}
\newcommand{\B}{\ensuremath{\mathcal{B}}}
\newcommand{\C}{\ensuremath{\mathcal{C}}}
\newcommand{\D}{\ensuremath{\mathcal{D}}}

\newcommand{\Hil}{\ensuremath{\mathcal{H}}}
\newcommand{\JZ}{\ensuremath{\mathcal{Z}}}
\newcommand{\KA}{\ensuremath{{\mathbb A}}}

\newcommand{\KC}{\ensuremath{{\mathbb C}}}
\newcommand{\KD}{\ensuremath{{\mathcal D}}}

\newcommand{\KG}{\ensuremath{{\mathcal G}}}
\newcommand{\KH}{\ensuremath{{\mathcal H}}}
\newcommand{\KR}{\ensuremath{{\mathbb R}}}
\newcommand{\KZ}{\ensuremath{{\mathbb Z}}}
\newcommand{\KQ}{\ensuremath{{\mathbb Q}}}

\newcommand{\KT}{\ensuremath{{\mathbb T}}}
\newcommand{\KO}{\ensuremath{{\mathcal O}}}

\newcommand{\U}{\ensuremath{{\mathcal{U}}}}
\newcommand{\KW}{\ensuremath{{\mathcal{W}}}}
\newcommand{\Aut}{\operatorname{Aut}}
\newcommand{\Ad}{\operatorname{Ad}}
\newcommand{\Br}{\operatorname{Br}}
\newcommand{\Bun}{\operatorname{\ensuremath{\mathpzc{Bun}}}}
\newcommand{\Cpct}{\ensuremath{\mathcal{K}}}

\newcommand{\CrPr}[3]{ \ensuremath{#1 \underset{#3}{\rtimes} #2} }

\newcommand{\Ind}{\operatorname{Ind}}

\newcommand{\into}{\ensuremath{\hookrightarrow}}

\newcommand{\KK}{\operatorname{KK}}
\newcommand{\id}{\operatorname{id}}

\newcommand{\Prim}{\operatorname{Prim}}

\newcommand{\Tor}{\operatorname{Tor}}
\newcommand{\ch}{\operatorname{ch}}

\numberwithin{equation}{section}
\theoremstyle{plain}
\newtheorem{theorem}{Theorem}[section]
\newtheorem{lemma}[theorem]{Lemma}

\newtheorem{definition}{Def}[section]

\AddToHook{begindocument/before}{\RequirePackage{nameref}}

\begin{document}

\title[An Extension of Topological T-duality]{Topological T-duality For Bundles Of Strongly Self-Absorbing $C^{\ast}$-Algebras And Some Physical Applications}

\author[Ashwin S. Pande]{Ashwin S. Pande}

\begin{abstract}
We extend the $C^{\ast}-$algebraic formalism of Topological T-duality
to section algebras of locally trivial bundles
of strongly self-absorbing $C^{\ast}-$algebras and to a larger class
of String Theoretic dualities.
We argue that physically this corresponds extending Topological T-duality to
Flux Backgrounds of Type II String Theory which possess topologically
nontrivial sourceless Ramond-Ramond flux. We demonstrate a map in $K-$theory
for the $C^{\ast}-$algebras involved in both sides of this generalized duality.
We calculate a few examples. We discuss the physical relevance of the above formalism in
some detail, in particular, we argue that the above formalism models String
Theoretic tree-level dualities found in such Flux Backgrounds such as Fermionic T-duality and Timelike T-duality.  
\end{abstract}

\maketitle

\section{Introduction\label{SecIntro}}

\subsection{Type II T-duality and Continuous-Trace $C^{\ast}-$algebras\label{S2T2CT}}

Consider Type IIA or Type IIB String Theory propagating on
a spacetime background $X.$ Suppose there is a 
free circle action $\alpha: S^1 \times X \to X$ on $X$
and $W = X/S^1.$
Then the quotient map $p:X \to W$ gives $X$ the structure
of a principal circle bundle over a base space $W.$ 

Type IIA and Type IIB String Theories propagating on these backgrounds
with sourceless $H-$flux  posesses a symmetry termed `T-duality' (see Ref.\ \cite{JMRCBMS} for
a good introduction). Roughly speaking, Type IIA String Theory on
the spacetime background with underlying topological space $X$ and
$H-$flux $[H] \in H^3(X,\KZ)$ is exactly
equivalent to Type IIB String theory on another `dual'
spacetime background with underlying dual topological space
$X^{\#}$ with dual $H-$flux $[H^{\#}] \in H^3(X^{\#}, \KZ).$
It can be shown that the dual spacetime $X^{\#}$
has a dual circle action $\alpha^{\#}: S^1 \times X^{\#} \to X^{\#}.$
In addtion $X^{\#}$ is also a principal circle bundle $p^{\#}:X^{\#} \to W$ 
due to the dual quotient map by the dual circle action $\alpha^{\#}.$

Topological T-duality is an attempt to understand the
T-duality symmetry of Type II String Theory on the target spacetime $X$
while retaining only the purely topological information about the string theory.
(It is unclear as to whether a similar approach will work for the other
known types of String Theories).
In Topological T-duality, the only information about type II string theory 
that we retain are the topological 
quantities associated to the zero energy modes of Type II
(A or B) String Theory propagating on the background $X.$ 
It is surprising that just these topological quantities alone are enough to construct
a mathematically rigorous, purely 
topological theory of Type II T-duality termed Topological T-duality (see the monograph
Ref.\ \cite{JMRCBMS} by J. Rosenberg for a detailed exposition).

In Type II String Theory propagating on a spacetime background $X,$ the massless
modes in the the $NS-NS$ sector are the
$B-$field-which corresponds to the data of a smooth
gerbe with connection on $X$-the metric and dilaton.

In topological T-duality (\cite{JMRCBMS}) only the following purely
topological information about the background $X$ are
retained: The characteristic class of the curvature $H$ of the $B-$field
together with the characteristic class of the principal circle bundle 
$p: X \to W.$
The characteristic class of the curvature $H$ is an integral 
three-form on $X$ and we identify it with a class $[H]\in H^3(X,\KZ)$ in the following.
The characteristic class of the principal circle bundle $p:X \to W$ is naturally
an element $[p] \in H^2(W,\KZ).$

It is remarkable that with such sparse information about the string background,
the formalism of Topological T-duality always correctly calculates the
topological type of the String Theoretic T-dual spacetime 
(see Ref.\ \cite{MRCMP, JMRCBMS} and references therein).
\subsection{Ramond-Ramond Fields}
Type II String Theory on a spacetime background $X$ possesses excitations termed
Ramond-Ramond fields whose field strengths are differential $p-$form valued fields (denoted $G_p$) on $X.$
It is well known that Type IIA String Theory has even-dimensional forms $$G_0, \ldots, G_{2k}, \ldots$$ as
Ramond-Ramond field strengths, while Type IIB String theory has odd-dimensional forms 
$$G_1, \ldots, G_{2k+1}, \ldots$$ as Ramond-Ramond field strengths.

We usually write the total field strength $G$ of the Ramond-Ramond fields
as an inhomogenous  sum of field strength forms
$$G = \Sigma_{i = 0}^5 G_{2i}$$ for Type IIA String Theory and
$$G = \Sigma_{i = 1}^4 G_{2i+1}$$ for Type IIB String Theory
(see Ref.\ \cite{Sati} Sec.\ (1)).

For any value of $p$ above, the field strength $G_{p}$ is the
curvature of a gauge potential $C_{p-1},$ i.e., 
$$G_{p} = dC_{p-1}.$$ 
$G_p$ has a gauge invariance 
$$C_{p-1} \rightarrow C_{p-1} + dB_{p-2}$$
with $B_{p-2}$ a $(p-2)$-form gauge field (The
$(p-2)-$form field $B_{p-2}$ should not to be 
confused with the Kalb-Ramond field $B$ which is
always a $2-$form field).
$D-$branes are sources of Ramond-Ramond fields.

The above is only an approximation, since it is known (see Ref.\ \cite{MW} Sec.\ (2), 
discussion around Eq.\ (2.19)) that Ramond-Ramond fields may be described by a hierarchy 
of successively finer descriptions. 

At the most basic level, the Ramond-Ramond field strengths 
$G_i$ on $X$ are closed forms on $X$ with the gauge symmetry described
above and should define natural classes in the de Rham cohomology of $X$. 
Most of these forms correspond to integral cohomology classes, but it may be 
argued from anomaly cancellation arguments (see Ref.\ \cite{MW}) 
that some of these classes are nonintegral. Hence,  this level of description is incomplete.

At a finer level of description, in Ref.\ \cite{MW}
Moore and Witten argued that Ramond-Ramond fields define natural 
classes in the $K-$theory group of $X,$ and that
the above field strength forms $G_i$ may be obtained from the associated 
$K-$theory class by the Minasian-Moore formula (see \cite{MW} Eq.\ (2.17)).

At a still finer level, Freed and Hopkins (see Ref.\ \cite{FH}) argued that
Ramond-Ramond fields on $X$ define natural classes in what is termed the
differential $K-$theory group of the spacetime background $X$
(see also Ref.\ \cite{Sati} Sec.\ (1) for an overview). 

We now give a brief description of the Minasian-Moore formula and its generalizations
for flux backgrounds:
If the spacetime has no background sourceless $B-$field
Witten and Moore (see Ref.\ \cite{MW, FH}) showed that
Ramond-Ramond fields naturally define elements of
the $K-$theory groups of spacetime.
Witten and Moore argued in Ref.\ \cite{MW} that for a Ramond-Ramond field
represented by a class $x \in K(X)$ the characteristic class $\omega$ of the Ramond-Ramond field 
associated to $x$ must be equal to the image of $x$ in the de Rham 
cohomology of $X$ via the Chern character, i.e., $\ch(x) = [\omega].$ Note that $\omega$ is a 
sum of differential forms of varying degrees. (For the connection of the above with the differential $K-$theory
group of $X,$ see Eqs. (10,11) of Ref.\ \cite{FH}.)

The inhomogeneous differential form $G$ representing the Ramond-Ramond field
is obtained from $\omega$ above by the Minasian-Moore formula 
\begin{equation}
G = \sqrt{\hat{A}(X)} \omega.
\end{equation}
where, as above, $G = \Sigma_i G_i$ where $G_i$ are the even or odd dimensional Ramond-Ramond field strength forms 
$G_i$ on $X$ above. Here, $\hat{A}(X)$ is the $\hat{A}-$genus of $X.$
While $\omega$ is an integral form, from the above formula, 
$G$ (since its obtained by multiplying $\omega$ by the form $\sqrt{\hat{A}(X)}$ which
might give rise to nonintegral forms) might not be integral (see Ref.\ \cite{MW}). 
Thus the Minasian-Moore formula gives an explanation of the nonintegral forms
in $G$ which were obtained from anomaly cancellation arguments.

If the spacetime has a background $B-$field with or without $H-$flux,
it has been argued that by various authors (see Refs.\ \cite{MS, Sati, Brodzki2})
that a generalization of the above argument should hold-in particular, we expect the $B-$field
to 'twist' various quantities. In Ref.\ \cite{MoSau}, Moore and Saulina show that for a flux
background with a $B-$field without $H-$flux there is an analogue of the 
Minasian-Moore formula twisted by the $B-$field: 
\begin{equation}
G(x)  = e^{B_2} \ch(X) \sqrt{\hat(A)(X)}
\end{equation}

When the background $B-$field possesses a nontrivial $H-$flux then 
Mathai and Sati in Ref.\ \cite{MS} and Brodzki et al in Ref.\ \cite{Brodzki2} argue that
the Ramond-Ramond fields satisfy a twisted Bianchi identity and are no longer closed under 
$d$ but vanish under a 'twisted differential' $d_H = d + H_3.$ The authors define a $H-$twisted
de Rham cohomology $H^{\ast}_H(X)$ using this twisted differential. 
It can also be shown (see Ref.\ \cite{Sati} and references therein), that there is a twisted
Chern character map $ch_H:K^{\ast}_H(X) \to H^{\ast}_H(X)$ where $K^{\ast}_H(X)$ is 
the $K-$theory of $X$ twisted by the closed, integral three form $[H]$ naturally 
defined by the $H-$flux $H$ in $H^3(X,\KZ).$ 

The authors argue that Ramond-Ramond fields define a class $y$ in $K^{\ast}_H(X)$ and the characteristic class of the Ramond-Ramond field in this flux background
is the image of $y$ via the twisted Chern character $\ch_H: K^{\ast}_H(X) \to H^{\ast}_H(X).$ 
In addition the authors argue that we should have $\ch_H(y) = [\omega]$ with a twisted version of the Minasian-Moore formula
\begin{equation}
G = \sqrt{\hat{A}(X)} \omega
\end{equation}
(For the generalization of the Minasian-Moore formula to twisted versions of differential $K-$theory
see Ref.\ \cite{Sati} discussion around Eqs. (1.7,1.8) and references therein).

Brodzki et al in Ref.\ \cite{Brodzki2} define a noncommutative version of the above for spacetime backgrounds which are `noncommutative manifolds', and in particular, derive a noncommutative Minasian-Moore formula. In Ref.\ \cite{Brodzki3}, Sec.\ (4.4) the authors show that the above twisted Minasian-Moore formula for backgrounds with sourceless $H-$flux agrees with the result in the previous paragraph for noncommutative manifolds which are described by continuous-trace $C^{\ast}-$algebras with spectrum $X$ (These were used to describe spacetime backgrounds with sourceless $H-$flux in Ref.\ \cite{MRCMP, JMRCBMS}).

In the rest of this paper, we argue that for spacetimes with a background sourceless $H-$flux 
and background sourceless Ramond-Ramond flux, the background Ramond-Ramond fields should be represented by a class $z \in K_{Tw}(X)$ the generalized twisted $K-$theory of $X$ in the sense of Ref.\ \cite{DadarlatK}. The characteristic classes of the Ramond-Ramond fields above should be related to the image of $z$ via a generalized twisted Chern character $\ch_{Tw}$ taking values in the generalized twisted cohomology of $X.$  

We use the results of Refs.\ \cite{DP1, DP2} to identify these characteristic classes with the characteristic classes of generalized twisted gerbes on the space $X$ and further identify these with the generalized Dixmier-Douady characteristic classes of the $C^{\ast}-$algebras of Ref.\ \cite{DP1, DP2}.
We study the effect of T-duality on these backgrounds. We define an extension of the
$C^{\ast}-$algebraic formalism of Topological T-duality to these $C^{\ast}-$algebras in
Secs.\ (\ref{SecLift}, \ref{SecMathEx}) below. We study some physical examples of this
in Sec.\ (\ref{SecPhys}) below.

We further argue below that the noncommutative Minasian-Moore formula of Brodzki et al
in Ref.\ \cite{Brodzki2} should describe some of the backgrounds in this paper in
Subsec.\ (\ref{S2TDPurInf}).

\subsection{String Theory in Flux Backgrounds}
In this section we discuss Type II String Theory in backgrounds with sourceless
$H-$flux or Ramond-Ramond flux-as mentioned in the previous section
these are termed {\em Flux Backgrounds}.

It was originally suspected that Type II String Theory was inconsistent in
backgrounds with Ramond-Ramond flux. However, even though these
backgrounds were suspected of being inconsistent,
nevertheless they were studied using String Field Theory techniques
beginning with the work of Bernstein and Leigh (see the references in Ref.\ \cite{BL}).
Over the past few years, the remarkable work of Sen and collaborators (see Ref.\ \cite{Sen} ) 
has shown that Type II String Theory is stable in these backgrounds: Type II Superstring Field 
Theory equations of motion exist in these backgrounds, but can't be obtained from an action. However,
the authors obtain gauge invariant nonlocal 1PI equations of motion from String Field Theory
and from this they obtain the full $S-$matrix of the String
Theory including all quantum corrections.

$D-$branes in Flux backgrounds have also been studied (see Refs.\ \cite{ CCS} for example):
For example, in Ref.\cite{CCS},  Cornalba et al argue that $D-$branes in flux backgrounds can be studied
using String Field Theory. They derive a number of interesting results from this which we will discuss below.

In this paper we attempt to study $D-$brane charge in flux backgrounds by extending
(or deforming) the $C^{\ast}-$algebraic formalism of Topological T-duality mentioned in
the previous section. 

Using a result of Cornalba et al. (in Ref.\ \cite{CCS}), we argue that just as switching 
on the $H-$flux on a Type II String Background causes the
spacetime to become noncommutative, similarly, switching on the Ramond-Ramond flux should cause
the spacetime to become a noncommutative manifold in the sense of Brodzki et al (Ref.\ \cite{Brodzki1,Brodzki2,Brodzki3}).
We claim that analogously with the formalism of Mathai and Rosenberg this noncommutative manifold may be described by the section algebra of a locally trivial $C^{\ast}-$bundle over the spacetime background with fiber a stabilized strongly self-absorbing $C^{\ast}-$algebra  $\A \otimes \Cpct$
(Note that the formalism of Mathai and Rosenberg is
the case with fiber $\KC \otimes \Cpct$). We propose that changes in the fiber $C^{\ast}-$algebra
reflect changes in the String Theory phenomenology.
We argue that a generalization the `Axiomatic T-duality' of Sec.\ (3.1) of Ref.\ \cite{Brodzki1} 
should hold for such backgrounds. We conjecture that the $D-$brane charge formula of Ref.\ \cite{Brodzki3} 
should be valid for these backgrounds. 

We define a generalized noncommutative Topological T-dual of the spaces discussed in this paper.
We argue that the formalism we propose displays a generalization of the Axiomatic T-duality of Brodzki et al.
We relate this generalized Topological T-dual to two phenomenological Tree-Level dualities in these
backgrounds-Fermionic T-duality and Timelike T-duality.

\subsection{The $C^{\ast}-$algebraic formalism of Topological T-duality}
In Ref.\ \cite{MRCMP}, Mathai and Rosenberg
developed the theory of Topological T-duality using ideas from
noncommutative geometry based on the crossed product of 
continuous-trace $C^{\ast}-$algebras by $\KR^n$-actions.

Continuous-trace $C^{\ast}$-algebras are a very interesting class 
of $C^{\ast}$-algebras which have been studied for a long time. 
Briefly, a continuous-trace $C^{\ast}-$algebra $\B$ over a compact metrizable space $X$ 
is a section algebra of a locally trivial $C^{\ast}$-bundle over
$X$ with fiber the $C^{\ast}$-algebra of compact operators 
$\Cpct$ on a countably infinite dimensional, separable Hilbert space.
The space $X$ may be recovered from $\B$ as the spectrum of
$\B.$ The set of isomorphism classes of continuous-trace algebras
with spectrum $X$ forms a group under $C_0(X)-$balanced tensor product of
continuous-trace algebras called the Brauer Group of $X$ denoted
$\Br(X).$ It can be proved (see Ref.\ \cite{JMRCBMS} and references therein) that
given a continuous-trace $C^{\ast}-$algebra $\B$ with spectrum $X,$
there is an isomorphism $\delta:\Br(X) \simeq H^3(X,\KZ)$ 
and the image $\delta([\B])$ of $[\B] \in \Br(X)$ 
under the above isomorphism is termed the Dixmier-Douady invariant
of the continuous-trace $C^{\ast}-$algebra $\B.$ A continuous-trace
$C^{\ast}-$algebra with spectrum $X$ and Dixmier-Douady invariant
$\delta$ is denoted $CT(X,\delta).$

It was argued in Ref.\ \cite{MRCMP} 
that, from the point of view of Noncommutative Geometry,
the presence of a $H-$flux $H$ with characteristic class
$[H] \in H^3(X,\KZ)$ on the spacetime background
$X$ corresponds to  replacing $C_0(X),$ the algebra of functions on $X$ when there is 
no $H-$flux by a continuous-trace algebra $\B$ with spectrum $X$ with
the Dixmier-Douady invariant of $\B$ equal to the $H-$flux $[H],$  
$\delta([\B]) = [H]-$see Ref.\ \cite{JMRCBMS} for details.
Thus, we may say that turning on a $H-$flux with characteristic class
$[H]$ causes us to replace $C_0(X)$ by $\B = CT(X, [H])$. 

The space $X$ possesses a circle action (as explained in 
Subsec.\ \ref{S2T2CT}). The Topological T-dual that Mathai and
Rosenberg define (see Ref.\ \cite{MRCMP}) depends on 
lifting the circle action $\alpha$ on $X$ to an action of 
$\KR$ on $\A$ covering the circle action on $X.$

It can be shown (see Ref.\ \cite{MRCMP, RaeRos,JMRCBMS})
that the $S^1-$action $\alpha$ on $X$ lifts to 
a unique equivalence class (termed exterior equivalence class, see Sec.\ \ref{SecLift})
of $\KR-$actions $\gamma$ on $\A$
each of which induces the given $S^1-$action on $X.$
This question of existence and uniqueness of lifts of group
actions up to exterior equivalence is part of the study of the crossed product of
continuous-trace $C^{\ast}-$algebras which is a vast and
well-developed body of work, see for example Refs. \cite{JMRCBMS,WillCP,RaeRos} 
for an overview of the subject.

Mathai and Rosenberg argue that the T-duality
operation on $X$ is modeled well by the crossed product of $\A$ by the natural lift of the
circle action on $X$ to an $\KR-$action $\gamma$ on $\A.$ 
More precisely, the T-dual space to $X$ (termed $X^{\#}$) is given by the spectrum of the
$C^{\ast}-$algebra $\A^{\#}$ which is the crossed product $C^{\ast}-$algebra
of $\A$ by the $\KR-$action $\gamma,$ i.e., $\A^{\#} \simeq \CrPr{\A}{\KR}{\gamma}.$ 
Varying $\gamma$ in its exterior equivalence class gives isomorphic crossed products.

Note that the topological type of the spectrum of $\A$ may be described
in terms of the topological data $[p], [H]$ above associated to the Type II String Theory background
with underlying topological space $X$. 
Similarly, the topological type of the spectrum of $\A^{\#}$ may be described in terms of the
topological data $[p^{\#}], [H^{\#}]$ associated to the dual Type II String Theory background
with underlying topological space $X^{\#}$

The $C^{\ast}$-algebraic approach to Topological T-duality 
(see Refs.\ \cite{MRCMP, JMRCBMS}) postulates that
a spacetime background $X$ with a {\em free} circle action and a {\em sourceless} $H-$flux
$[H]$ (denoted $(X,[H])$ in what follows) can be associated to an isomorphism class of $C^{\ast}-$dynamical
systems $[\A,\alpha].$ Here $\A$ is a continuous-trace algebra
with spectrum $X$ and Dixmier-Douady invariant $[H]$ 
and $\alpha$ is a lift of the circle action on $X$
to a $\KR-$action on $\A.$ 
In this theory, the Dixmier-Douady invariant of the continuous-trace algebra $\A$
is fixed to be equal to $[H],$ the class in $H^3(X,\KZ)$ induced by the sourceless $H-$flux 
$H.$

This is natural physically because of the following argument:  The 
charges of the $D-$branes in the
string background $X$ in the absence of $H-$flux 
lie in the the topological $K-$theory group of spacetime.
It is conjectured (see Ref.\  \cite{JMRCBMS} Sec. (4.2.08) ) that if a {\em sourceless} $H-$flux $H$ were switched
on in the spacetime $X$, the $D-$brane charges would lie in the twisted topological
$K-$theory group $K_H(X)$ where the magnitude of the twist is given by the value of
the $H-$flux. 

In this theory, the $K-$theory of the continuous-trace $C^{\ast}-$algebra $\A$ is
also a receptacle for $D$-brane charges: The twisted $K$-theory groups 
of the background (which are conjectured to be the charge group of $D-$branes in 
backgrounds with sourceless $H-$flux) are isomorphic to the operator $K$-theory groups 
of a continuous-trace $C^{\ast}$-algebra on that background (see
for example, Ref.\ \cite{JMRCBMS} Sec.\ (4.2) and references therein). 

Further, it was proved  in Ref.\ \cite{BouwPan} that the
$C^{\ast}$-dynamical system associated to a principal bundle
$p:W \times S^1 \to W$ with $H-$flux by the formalism of Ref.\ \cite{MRCMP} may also
be naturally constructed from the data of the string theory background.

Thus, this association of a continuous-trace $C^{\ast}-$algebra and associated
$C^{\ast}-$dynamical system to the spacetime background $(X,[H])$ 
is physically very natural since it is constructed from both
nonperturbative data-the $D-$brane charge group and perturbative data-the equivariant
gerbe on spacetime whose gerbe curvature is the characteristic class $[H]$ of the
sourceless $H-$flux.

It was observed in Ref.\ \cite{MRCMP} that with the above setup, 
the T-duality operation corresponds to the crossed-product
construction in $C^{\ast}$-algebra theory and the underlying
topological space to the T-dual is the 
spectrum of $\CrPr{\A}{\KR}{\alpha}.$ Thus, the formalism
of Ref.\ \cite{MRCMP} associates the $C^{\ast}-$dynamical system 
$[ \CrPr{\A}{\KR}{\alpha}, \alpha^{\#}]$ to the T-dual of $[\A,\alpha].$
 
Under T-duality, $D$-branes map by a change of degree-i.e. branes of even
degrees go to branes of odd degrees. This phenomenon is visible in the
above formalism as explained in Ref.\ \cite{MRCMP}. 
The mapping of $D-$brane charges under T-duality (see Ref.\ \cite{JMRCBMS})
was observed to correspond to the Connes-Thom isomorphism in $C^{\ast}-$algebra
theory $\phi: K_{\bullet}(\A) \to K_{\bullet + 1}(\CrPr{\A}{\KR}{\alpha}).$

The formalism of Mathai and Rosenberg
has been directly tested in many examples (see Ref.\ \cite{MRCMP}
and references therein; see also Ref.\ \cite{JMRCBMS}). It remarkable
that the formalism has always been found to correctly give the underlying
topological space of the String Theoretic T-dual spacetime.

\subsection{ Topological T-duality for Flux Backgrounds}

We would like to extend Topological T-duality to the study of
dualities in String Theory on backgrounds with sourceless Ramond-Ramond
Flux and $H-$flux. To do this, we need to use a generalization of notion of 
a continuous-trace $C^{\ast}-$algebra over $X$ to the notion of a section
algebra of a locally trivial bundle of strongly self-absorbing 
$C^{\ast}$-algebras over $X.$

A $C^{\ast}-$algebra $\A$ is said to be strongly self-absorbing if it is separable
unital and there is a $\ast-$isomorphism $\psi:\A \to \A \otimes \A$ such that
$\psi$ is approximately unitarily equivalent to $l:\A \to \A \otimes \A, l(a) = a \otimes 1,$
i.e., there is a sequence of unitaries $v_n \in \A \otimes \A$ such that
$\forall x \in \A,$ 
$$||v_n \psi(x) v_n - (x \otimes 1)|| \to 0$$ as $n \to \infty.$
These algebras are important in the Elliott programme and arise naturally
there in an attempt to classify isomorphism classes of $C^{\ast}-$algebras
using invariants from $C^{\ast}-$algebraic $K-$theory.

There are six known types of strongly self-absorbing $C^{\ast}-$ algebras.
These are listed in Table (\ref{tabSSA})-Note that $\KW,$ the Razak-Jacelon algebra
(see Ref.\ \cite{J1}) is nonunital and hence does not fit the above definition of a strongly self-absorbing $C^{\ast}-$algebra but will be used to calculate an example later.

\begin{table}[th]
\begin{center}
{
\caption{\label{tabSSA}This is a list of all known strongly self-absorbing $C^{\ast}-$algebras adapted from Ref.\ \cite{J1}, note that $\KW$ is not unital.}
\begin{tabular}{c c ||| c  c}
\hline 
 Remarks & {\bf Stably finite} & {\bf Purely Infinite} &  Remarks \\ [0.5ex]
\hline\hline
{ Nawata et al Ref.\ \cite{N1}} & $\KW$ & $\KO_2$ & \rdelim\}{3}{1mm}[{ ${\KO_{\infty}-}$absorbing}]  \\
{Described Below} & $M_{p^{\infty}}$ & $M_{p^{\infty}} \otimes \KO_{\infty}$\\
{Sato et al Ref.\ \cite{Sato1}} & $\JZ$ & $\KO_{\infty}$ \\
\hline 
\end{tabular}
}
\end{center}
\end{table}

From Ref.\ \cite{TomsWin}, every strongly self-absorbing $C^{\ast}-$algebra
is either stably finite or purely infinite.

First we note the following: A $C^{\ast}-$algebra $\KD$ as above is said to be
{\em $\B-$absorbing} for some strongly self-absorbing $C^{\ast}-$algebra $\B,$ 
if $\KD \otimes \B \simeq \KD.$ Note that if $\A$ is $\B-$absorbing, so is
$\KD \in \Bun_X(\A \otimes \Cpct).$  In particular, every strongly self-absorbing $C^{\ast}-$algebra $\A$ is always Jiang-Su absorbing,
so $\KD$ above is always $\JZ-$absorbing.

There are other absorption results which depend on the specific choice of $\A$. 
For example, if $\A = \KO_2,$ we have $\KO_2 \otimes M_{2^{\infty}} \simeq \KO_2$
and so every $\KD$ in $\Bun_{X}(\KO_2 \otimes \Cpct)$ is $M_{2^{\infty}}-$absorbing.

Dadarlat and Pennig in Ref.\ \cite{Dadarlat} considered section algebras of locally trivial fiber bundles 
with fiber a (stabilized) strongly self-absorbing $C^{\ast}-$algebra $\A \otimes K$
over a compact metrizable space $X.$ The isomorphism classes of these $C^{\ast}-$algebras over $X$ are denoted
$\Bun_X(\A \otimes \Cpct).$
In Ref.\ \cite{Dadarlat}, Thm.\ (3.8), the authors show that there is a generalized cohomology theory on the
category of $CW-$complexes $E^{\ast}_{\A}(X)$ such that
\begin{gather}
E^1_{\A}(X) = \Bun_X(\A \otimes \Cpct).
\label{SSACl1}
\end{gather}
In addition, in Cor.\ (4.3) of Ref.\ \cite{Dadarlat}, the authors calculate $\Bun_X(\A \otimes \Cpct)$
for a finite connected $CW-$complex $X$ using the Atiyah-Hirzebruch spectral sequence for 
$E^{\ast}_{\A}(X)$ when $\A \neq \KC$ and $\A$ satisfies the UCT. For such spaces $X$ and 
$C^{\ast}-$algebras $\A,$ the authors obtain that
\begin{gather}
\Bun_X(\A \otimes \Cpct) \simeq E^1_{\A}(X) \simeq H^1(X, R^{\times}_{+}) \times \Pi_{k \geq 1}
H^{2k+1}(X,R) \label{SSACl2}
\end{gather}

Note that if $\A \simeq \KC,$ 
$\Bun_X(\KC \otimes \Cpct) \simeq \Br(X) \simeq H^3(X,\KZ)$ the Brauer Group
of Morita equivalence classes of continuous-trace $C^{\ast}-$algebras with spectrum $X.$

In Sec.\ (\ref{SecPhys}) of this paper we claim that $X$ should naturally be viewed as the
underlying topological space of a Type II String Theory background
with various sourceless fluxes turned on-these are termed {\em flux backgrounds}
in String Theory. These sourceless fluxes are zero modes of the
Type II theory's spectrum like the $H-$flux or the Ramond-Ramond Flux.

The case $\A = \KC$ corresponds to the association of 
a continuous-trace $C^{\ast}-$algebra to a Type II String Theory background
with a sourceless $H-$flux as described in Refs.\ \cite{MRCMP, JMRCBMS}.

When $\A \neq \KC$ however, we claim the resulting $C^{\ast}-$algebra $\KD$ still describes
a spacetime background, but we should associate associate locally trivial bundles with
fibers of the form $\A \otimes \Cpct$ with $\A$ as above to spacetime backgrounds with
{\em both} $H-$flux and $RR-$flux. 

We show below that the formalism of Topological T-duality for continuous-trace
$C^{\ast}-$algebras can be generalized from $\A \simeq \KC$ to any 
self-absorbing $C^{\ast}-$algebra $\A.$
One difficulty with this generalization is the lack of a theory of $\KR-$actions 
on the above class of $C^{\ast}-$algebras. Hence, the structure of the crossed product
by even relatively simple groups such as $\KR$ cannot be determined. 
For continuous-trace $C^{\ast}-$algebras there is a satisfactory theory of $\KR-$actions
which lift circle actions on the spectrum: The structure of the crossed-product 
$C^{\ast}-$algebra by these $\KR-$actions is known precisely, and it is always clear when the crossed-product $C^{\ast}-$algebra is a continuous-trace
$C^{\ast}-$algebra at least for elementary group actions such as $\KR^n$ (see Refs.\ \cite{RaeRos, JMRCBMS, MRCMP}).
This is not the case for the class of $C^{\ast}-$algebras studied in this paper.

We outline this theory in the rest of this paper. 
We discuss the lifting of circle actions on $X$ to $\KR-$actions on $\KD$ in Sec.\ (\ref{SecLift}) below and show that unique `Rokhlin' lifts (up to closed cocycle conjugacy-a generalization of
exterior equivalence) of the circle action on $X$ to $\KR-$actions on $\KD$ exist for purely infinite strongly self-absorbing fibers $\A$-namely the Cuntz algebras $\KO_2, \KO_{\infty}$ and the tensor products  of $\KO_{\infty}$ with UHF algebras of infinite type 
$M_{p^{\infty}} \otimes \KO_{\infty}.$ 

When $\A$ is stably finite that is, $\A$ is any of the three remaining strongly self-absorbing 
$C^{\ast}-$algebras-the Razak-Jacelon algebra $\KW$, the UHF-algebras of 
infinite type $M_{p^{\infty}}$ or the Jiang-Su algebra $\JZ$
there does not seem to be a unique lift of the circle action on $X$ to a $\KR-$action on $\KD.$

We define a generalization of the theory of Topological T-duality for both these classes of $C^{\ast}-$algebras in Sec.\ (\ref{SecLift}) below and argue that this theory gives the expected T-dual in the continuous-trace case i.e., when $\A \simeq \KC.$
We calculate the Topological T-dual of several spaces in Sec.\ (\ref{SecMathEx}) below.
We discuss the implications for String Theory and $D-$brane charge in Flux backgrounds
and also calculate some examples of T-duals using the above in Sec.\ (\ref{SecPhys}) below.
We end the paper with our conclusions in Sec.\ (\ref{SecConclusion}).

\section{Lifts of circle actions \label{SecLift}}
In this section we argue about possible lifts of circle actions on the total space of a principal circle bundle
$p:X \to W$ to $\KR-$actions on $\D \in \Bun_{X}(\A \otimes \Cpct).$

Let $X$ be a topological space associated to a finite-dimensional manifold.
Suppose $p:X \to W$ was a principal circle bundle. Consider
$C_0(X)$-algebras $\D$ which are locally trivial $C^{\ast}$-bundles
over $X$ with fiber a strongly self-absorbing $C^{\ast}$-algebra
$\A \otimes \Cpct$ as in Ref.\ \cite{Dadarlat}. As in Ref.\ \cite{Dadarlat},
Thm.\ (A), we denote the isomorphism classes of these $C^{\ast}$-algebra bundles
over $X$ by $\Bun_{X}(\A \otimes \Cpct).$

In Subsec.\ (\ref{S2PrimD}) we calculate the structure of $\Prim(\KD)$ for $\KD$ above.
In Subsecs.\ (\ref{S2RAct}, \ref{S2ZS1Act}) we briefly discuss Rokhlin $\KR-$actions and 
Rokhlin $\KZ-$ and $S^1-$actions on the $C^{\ast}-$algebras $\KD$ above.
In Subsec.\ (\ref{S2TDPurInf}) we discuss the lifting of circle actions on $X$ to $\KR-$actions
on $C^{\ast}-$algebras $\KD$ above with purely infinite fiber and
Topological T-duality for these $C^{\ast}-$algebras.
In Subsec.\ (\ref{S2TDStaFin}) we discuss the lifting of circle actions on $X$ to $\KR-$actions
on $C^{\ast}-$algebras $\KD$ above with stably finite fibers and
Topological T-duality for these $C^{\ast}-$algebras.
In Subsec.\ (\ref{S2KTM}) we discuss maps on $K-$theory induced by the constructions
of Subsecs.\ (\ref{S2TDPurInf}, \ref{S2TDStaFin}).
In Subsec.\ (\ref{S2Gen}) we discuss some generalizations of the above.

\subsection{Structure of $\Prim(\KD)$ \label{S2PrimD}}

In this paper, to a $C^{\ast}-$algebra like $\KD$ as above,
we associate the topological space $\Prim(\KD)$ as the underlying spacetime background
of Type II String theory with background $H-$ and $RR-$flux. 
Also if $\Prim(\KD)$ does not exist, we view the spacetime background as the
noncommutative space given by $\KD.$

In this section, we calculate $\Prim(\KD)$ for 
$\D \in \Bun_{X}(\A \otimes \Cpct).$ The proof below has been adapted from 
Prop.\ (5.36) of I. Raeburn and D. Williams Ref.\ \cite{RaeWill}.
The proof there is for continuous-trace $C^{\ast}-$algebras and 
has been generalized to locally trivial $C^{\ast}-$bundles
with fiber a strongly self-absorbing $C^{\ast}-$algebra.

\begin{theorem}
Let $\A, X, \D$ be as at the beginning of this section.
For $\D \in \Bun_{X}(\A \otimes \Cpct),$ $\Prim(\D) \simeq X.$
\label{ThmDPrim}
\end{theorem}
\begin{proof}
First note that $\A,\Cpct$ are separable and simple and $\Cpct$ is nuclear. 
By Thm.\ (B.45) (c) of Ref.\ \cite{RaeWill}, $\Prim(\A \otimes \Cpct) 
\simeq \Prim(\A) \times \Prim(\Cpct) \simeq \{ \ast \}.$ Hence, the
kernel of any irreducible representation of $\A \otimes \Cpct$ must
be $\{ 0 \}.$

By Ref.\ \cite{Dadarlat}, $\D$ is a $C_0(X)$-algebra, hence, there is
an embedding $C_0(X) \into ZM(\D).$ If $\pi$ is an irreducible
representation of $\D$ on a Hilbert space $\Hil,$ 
let $\overline{\pi}$ denote the extension of
$\pi$ to $M(\D).$ Restrict $\overline{\pi}$ to $C_0(X) \subset M(\D).$
By the above, $\overline{\pi}(f) \overline{\pi}(a) = \overline{\pi}(fa) = 
\overline{\pi}(a) \overline{\pi}(f)$ and $\overline{\pi}(f)$ commutes with 
$\pi(a)$ for every $a \in \D.$ Since $\pi$ is irreducible, 
this implies that $\overline{\pi}(f) \in \KC 1_{\Hil}.$ Thus, we have an
irreducible representation of $C_0(X)$ on $\Hil$ which implies that
there is a $t \in X$ such that $\overline{\pi}(fa) = f(t) \pi(a),
\forall f \in C_0(X), a \in \D.$ 

Let $I_t$ denote the set of sections of the $C^{\ast}$-bundle
$\D$ which vanish at $t \in X.$ Then, by Ex.\ (A.23) of 
Ref.\ \cite{RaeWill}, second paragraph, we can argue that 
$I_t \subset \ker(\pi).$  Also, since $a \mapsto a(t)$ induces
an isomorphism of $\D/I_t$ onto $\D|_{\{t\}},$ $\pi$ factors through
a representation $\tilde{\pi}$ of $\D|_{\{t\}}.$ 
Since $\D|_{\{t\}} \simeq \A \otimes \Cpct,$ and 
$\Prim(\A \otimes \Cpct) \simeq \{0\},$ by the argument at the beginning
of the proof,we have that $\ker(\tilde{\pi}) \simeq \{ 0 \}.$ 
Hence, $\ker(\pi) = I_t.$

If $\eta$ is any irreducible representation of $\D|_{\{t\}} \simeq \A \otimes
\Cpct$ the evaluation map which sends a section $s \in \D$ to 
$\eta(s(t))$ gives a representation of $\D$ with kernel $I_t.$

Let $p:Y \to X$ be a bundle with fiber $\A \otimes \Cpct(\Hil)$
and structure group $\Aut(\A \otimes \Cpct(\Hil))$ over $X.$
By Ref.\ \cite{Dadarlat}, $\D$ is the section algebra of this
bundle. Now, the proof of Prop.\ (4.89) of Ref.\ \cite{RaeWill}, 
second paragraph, with $\Cpct(\Hil)$ replaced by 
$\A \otimes \Cpct(\Hil)$ shows that 
for each $y \in Y,$ there is a section $s$ of $\D$ with 
$s(p(y)) = t.$ 

Thus, as in Ref.\ \cite{RaeWill}, Prop.\ (5.36), first paragraph,
every element of each fiber is the value of some section.
We can also multiply by elements of $C_0(X).$
Hence, ideals $I_t$ corresponding to different $t \in X$ are 
distinct. Thus, $t \mapsto I_t$ is a bijection of $X$ onto $\Prim(\D).$

Pick an open cover $\{W_{\alpha}\}$ of $X$ such that for every
$W_{\alpha}$ there is a closed set $V_{\alpha} \supseteq W_{\alpha}$ such
that Thm.\ (B) of Ref.\ \cite{Dadarlat} lets us
pick a projection $p_{\alpha} \in \D|_{V_{\alpha}}.$
By Thm.\ (B) of Ref.\ \cite{Dadarlat}, $[p_{\alpha}] \in
K_0(\D|_{V_{\alpha}})^{\times}$ hence, $p_{\alpha}$ is not zero.

It is enough to show that the map $t \mapsto I_t$ is a homeomorphism.
The proof is similar to the one in the last paragraph of Ex.\ (A.23) of Ref.\ \cite{RaeWill}.
It is enough to show that 
$$
\overline{\{I_t : t \in N \}} = \{I_t : t \in \overline{N} \}
$$
for every subset $N$ of $X.$ By definition, the left hand side
has the form $\{ I_t: t \in M\}$ where, (by 
Def.\ (A.19) of Ref.\ \cite{RaeWill}),  
$M = \{ s:  I_s \supset \bigcap_{t \in N} I_t \},$
and we have to prove $M = \overline{N}.$
Since $\D$ is the algebra of continuous sections
of a locally trivial $C^{\ast}$-bundle over $X,$ 
for any $a \in \D,$ $a(x) = 0$ for $x \in N$ implies 
that $a(x) = 0$ for $x \in \overline{N}$ 
so $\overline{N} \subset M$ by definition of $M.$

If $s \notin \overline{N},$ there exists $f \in C_0(N),$ such
that $f|_{\overline{N}} = 0$ and $f(s) = 1.$ 
Pick an $\alpha$ such that $s \in W_{\alpha}.$
Then, by the above, there exists a nonzero projection 
$p_{\alpha} \in D|_{V_{\alpha}}.$ Now,
$fp_{\alpha}$ is in $I_t$ for all $t \in \overline{N}$
but not in $I_s$ which implies that $s \notin M.$ Thus,
$M = \overline{N}.$
\end{proof}

\subsection{Introduction to Rokhlin $\KR-$Actions \label{S2RAct}}
We begin by defining exterior equivalence of two group actions on a $C^{\ast}-$algebra and
a related notion called cocycle conjugacy in this Subsection.
Then we illustrate the idea of a Rokhlin $\KR-$action by examining Rokhlin $\KR-$actions
on a unital $C^{\ast}-$algebra. We indicate the generalization of this definition to nonunital
$C^{\ast}-$algebras. 

In this paper we will need to consider lifts of circle actions on $X \simeq \Prim(\KD)$ 
to what are termed {\em Rokhlin} $\KR-$ or $S^1-$actions on $\KD.$ 
We would also like to consider Rokhlin $\KZ-$actions on $\KD.$
More precisely, we consider $C^{\ast}$-dynamical systems of the 
form $(\D, \alpha, \KR)$ 
where $\alpha$ is a Rokhlin $\KR$-action on $\D$ which is a lift of the 
circle action on $X.$ We also consider actions of the form $(\KD, \beta, S^1)$ where
$\beta$ is a Rokhlin circle action on $\KD$ which is
a lift of the circle action on $X.$ We would also like to consider $C^{\ast}-$dynamical
systems of the form $(\KD, \gamma, \KZ)$ where $\gamma$ is a Rokhlin 
$\KZ-$action on $\KD.$

We first recall the notion of exterior equivalent actions and define cocycle conjugate actions:

\begin{definition}
Let $\KD, \KD'$ be two section algebras of locally trivial bundles of
strongly self-absorbing $C^{\ast}-$algebras with the same primitive
spectrum $X.$
\leavevmode
\begin{itemize}

\item Two actions $\alpha, \alpha'$ of a second countable, locally compact 
group $G$ on $\KD$ are {\em exterior equivalent}
if there is a continuous map $u:G \to UM(\KD),$ 
the group of unitary multipliers in the multiplier algebra of $\KD,$
with $u_{gh} = u_g \alpha_g(u_h)$ such that 
$\alpha'_g = \Ad u_g \circ \alpha_g.$

\item Two actions $\alpha, \beta$ of a second countable, locally compact group $G$
on $\KD$ and $\KD'$ are {\em cocycle conjugate} if there is an $\alpha-$cocycle 
$u$ i.e., a continuous map $u:G \to UM(\KD),$ the group of unitary multipliers
in the multiplier algebra of $\KD,$ with $u_{gh} = u_g \alpha_g(u_h)$ such
that $\alpha^u_g = \Ad u_g \circ \alpha_g$ is another action. 
In addition there must be an isomorphism $\phi: \KD \to \KD'$ such that 
$\alpha^u_g = \phi^{-1} \circ \beta_g \circ \phi, \forall g \in G.$
\end{itemize}
\end{definition}

We argue below that to obtain a T-dual we require a 
lift of the circle action on $X$ to a $\KR-$action on $\KD$ or to a
$S^1-$action on $\KD.$ 

Note that if $\A \simeq \KC,$ that is, $\KD$ is a continuous-trace $C^{\ast}-$ algebra
the circle action on the spectrum of $\KD$ lifts to a unique 
$\KR-$action (up to exterior equivalence) 
on $\KD$ but not to a nontrivial circle action on $\KD$ (see Ref.\ \cite{RaeRos}), 
however, this is not the case for general $\D$ in $\Bun_X(\A \otimes \Cpct)$ with
$\A \neq \KC.$ 

We cannot guarantee
a unique lift of circle action on $X$ to a $\KR-$action on $\KD$ without further conditions
on the lifted action.  Therefore, instead of lifting the circle action on $X$ to an 
arbitrary $\KR-$action, we lift it to a specific type of action
termed a {\em Rokhlin action} which we explain below.

It is suspected that there is no way to classify group actions on $\A-$absorbing
$C^{\ast}-$algebras with any sensible notion of
equivalence (including cocycle conjugacy above) unless the actions possess the Rokhlin property. That is a unique (in the sense of cocycle conjugacy) 
lift of the circle action on $X$
to a $\KR-$action on $\KD$ might not exist unless the lifted
$\KR-$action possesses the Rokhlin property (see \cite{S1} and \cite{K1} for more details).

Also for general $\KD,$ the fibers of $\KD$ are not $\Cpct$ but can be any stabilized
self-absorbing $C^{\ast}-$algebra and the lifting problem is difficult to
handle without further conditions on the fiber algebra $\A \otimes \Cpct.$ We divide
the lifting problem into two cases, those with $\A$ purely infinite and those with
$\A$ stably finite. 

Due to this, in this paper we do the following:
\begin{enumerate}
\item We principally restrict ourselves to {\em Rokhlin actions} of groups on $C^{\ast}-$algebras
and Rokhlin lifts of group actions on topological spaces to $C^{\ast}-$algebras $\KD$
in $\Bun_X(\A \otimes \Cpct)$ as far as possible (see Refs.\ \cite{K1, H1, HSWW}). 
\item In the absence of Rokhlin actions, in the following we also use actions possessing 
a weakening of the Rokhlin property called the {\em tracial Rokhlin property}
(see Ref.\ \cite{Phil} and references therein). 
\item If Rokhlin actions or actions with the tracial Rokhlin property cannot be proved to exist on 
$\KD,$ we will also use another generalization of Rokhlin actions  
(see Refs.\ \cite{H1, SWZ, G1} and references therein) to
actions with a {\em finite Rokhlin dimension}- a Rokhlin action has
Rokhlin dimension zero, Rokhlin actions with
finite Rokhlin dimension are a generalization of actions with Rokhlin dimension zero.
\item If none of these are possible for a given situation, we will use a natural group action
on $\KD,$ but we will point out that the action is not Rokhlin.
\end{enumerate}

The subject of Rokhlin actions on $C^{\ast}-$algebras is vast (see Refs.\ \cite{K1, H1, HirshWin, S1, HSWW}
and references therein) and cannot be discussed here in any detail.
We give below the simplest example of a Rokhlin $\KR-$action on a unital $C^{\ast}-$algebra
and refer the reader to the literature for more details and for generalizations.

The Rokhlin property for $\KR-$actions on a separable, unital $C^{\ast}-$algebra
(also termed {\em Rokhlin flows}) was first defined by Kishimoto in Ref.\ \cite{K1}. 
\begin{definition}
Let $\alpha$ be a flow on a separable, unital $C^{\ast}$-algebra $\A.$
We say that $\alpha$ has the Rokhlin property, if for every $p > 0,$ there exists an
approximately central sequence of unitaries $u_n$ in A satisfying
\begin{gather}
\lim_{n \to \infty} \max_{|t| \leq 1} || \alpha_t(u_n) - e^{ipt} u_n ||  = 0.
\end{gather}
\end{definition}

However, the $C^{\ast}-$algebras in this paper are always nonunital since $\Cpct$ is nonunital
and hence section algebras of locally trivial $C^{\ast}-$algebra bundles with fiber $
\A \otimes \Cpct$ are nonunital.  
This is not a problem, since there is a generalization of the idea of a Rokhlin flow to a flow on an
arbitrary separable (not necessarily unital) $C^{\ast}-$ algebra in Ref.\ \cite{S1}-see paragraph
before Def.\ (1.8) of Ref.\ \cite{S1}. 
We refer the interested reader to that reference.

We suggest two analogies which could give a physical interpretation of a Rokhlin $\KR-$action above.
\begin{itemize}
\item
Suppose that $X \simeq {\mbox{ pt } }$ so that $\KD \simeq \A \otimes \Cpct.$
Assume that $\A$ was the $C^{\ast}-$algebra of observables associated to a quantum field
theoretic system, and $\alpha_t$ was the time evolution induced by the Hamiltonian. In particular
assume that $\A$ was faithfully represented on the Fock space of a quantum field theoretic system.
The above definition would then imply that for some interval of time around
$t = 0,$ for every value of the one-particle momentum $p,$ there exist excitations in the system
described by $\\{ u_n, n = 1,2, \ldots \\}$ for every value of
$p$ which are `approximate eigenstates' of the Hamiltonian and of all other observables for a finite time.
In addition, it is clear that these excitations propagate as free particles,
at least for finite time. Such excitations are well known in Quantum Field Theory.

\item 
Alternatively, $X$ could be a background for a Type II String Theory.  Also, suppose the $\KR-$action 
on $\D$ covered the circle action on $X = \Prim(\D)$ as  above and the $C^{\ast}-$algebra $\D$ was
associated (by some means) to the algebra of observables associated to Strings propagating on $X.$ 
Then the Rokhlin condition would imply that in the Fock space there are orthonormal states which are approximate
eigenstates of every operator in $\D$ (since $b$ and $u_n$ approximately commute for every $b \in \D$)
which are nearly invariant under the translation action on $\D.$ These should correspond to states
of wound strings, at least in the limit $n \to \infty.$ It would be interesting to calculate this from String Theory.
\end{itemize}

\subsection{Rokhlin $\KZ-$ and $S^1-$actions \label{S2ZS1Act}}
In this subsection we discuss $\KZ-$ and $S^1-$actions 
on $\KD.$ We will use these in the next subsection, SubSec.\ (\ref{SecLift}).

If there is no Rokhlin lift of the circle action on $X$ to
to a Rokhlin $\KR-$action on $\D \in \Bun_{X}(\A \otimes \Cpct)$ 
then, we argue below that, in addition to the Rokhlin $\KR-$actions discussed above, 
we should consider Rokhlin circle actions and 
Rokhlin $\KZ-$actions or Rokhlin automorphisms and some crossed product $C^{\ast}-$algebras
by these actions. 

By the above discussion, these actions would have to be on not necessarily unital $C^{\ast}-$algebras. The reader is referred to Refs.\ \cite{G2,G3} for details on Rokhlin circle actions
and crossed products by Rokhlin circle actions on such $C^{\ast}-$algebras which are relevant to this paper. Note that Ref.\ \cite{G3} studies Rokhlin circle actions on $\sigma-$unital $C^{\ast}-$algebras which applies here since $\KD$ is separable.

We will also need some information about Roklin automorphisms on not necessarily unital
$\C^{\ast}-$algebras which are studied in Ref.\ \cite{SWZ}.
The Lemma below discusses the existence of Rokhlin $\KZ-$actions (also
called Rokhlin automorphisms or Rokhlin automorphisms of Rokhlin dimension 0)
on the $C^{\ast}-$algebras $\KD$ above. 

In Ref.\ \cite{H1}, Hirshberg, Zacharias and Winter introduce the idea of Rokhlin dimension for $\KZ-$actions on $C^{\ast}-$algebras.
Rokhlin automorphisms discussed above are considered to be $\KZ-$actions with Rokhlin dimension $0.$
There is a generalization of Rokhlin Automorphisms in to automorphisms which don't have the Rokhlin property
termed Rokhlin $\KZ-$actions of dimension greater than $0$ in Ref.\ \cite{SWZ}. The results in Ref.\ \cite{SWZ}
only guarantee the existence of Rokhlin $\KZ-$actions of dimension $\leq 1$ for arbitrary strongly self
absorbing $\A.$ Since crossed products by Rokhlin automorphisms of dimension greater than $0$ are understood,
if there are no results on Rokhlin automorphisms of dimension $0,$ we will take a crossed product by a Rokhlin
automorphism of dimension $\geq 0.$

However, if $\A$ absorbs the universal UHF-algebra there is a stronger result regarding
Rokhlin automorphisms of dimension $0$. 

In the following Lemma we prove that there are Rokhlin $\KZ-$actions on $\KD$ using the results of Ref.\ \cite{SWZ}
(Note that the results cited to prove 
Lemma (\ref{LemZActD}) don't need unitality of the $C^{\ast}-$algebra $\KD.$):
\begin{lemma}
Let $\A$ be a strongly self absorbing $C^{\ast}-$algebra.
Let $X$ be a compact topological space and $\KD$ a section algebra of a locally trival 
$C^{\ast}-$bundle over $X$ with fiber $\A \otimes \Cpct.$ \label{LemZActD}
\leavevmode
\begin{enumerate}
\item For any $C^{\ast}-$algebra $\KD$ as above, Rokhlin $\KZ-$actions on $\KD$ with Rokhlin dimension $\leq 1$ 
are generic in the set of $\KZ-$actions on $\KD$ (in the sense of Remark (10.2) of Ref.\ \cite{SWZ}) 
\item If $\A$ is absorbs the universal UHF algebra, then Rokhlin automorphisms of $\KD$ are generic in the set of $\KZ-$actions
on $\KD$ (in the sense of Remark (10.2) of Ref.\ \cite{SWZ})
\end{enumerate}
\end{lemma}
\begin{proof}
\leavevmode
\begin{enumerate}
\item
Since $\A$ is strongly self absorbing, it is Jiang-Su absorbing. 
Since $\KD$ is a section algebra of a $C^{\ast}-$bundle over $X$ with fiber $\A \otimes \Cpct$ it is clear that
$\KD$is Jiang-Su absorbing as well.
By Thm.\ (10.14) and Remark (10.15) of Ref.\ (\cite{SWZ}), since
$\KZ$ is a finitely generated, countable, discrete, residually finite group and $\KD$ is separable,
Rokhlin $\KZ-$actions are generic in the set of all actions of $\KZ$ on $\KD$ in the sense of Def.\ (10.1)
and Remark (10.2) of Ref.\ (\cite{SWZ}). 
\item If $\A$ absorbs the universal UHF algebra, then by an argument similar to the previous part 
$\KD$ absorbs the universal UHF algebra as well. By Thm. (10.10) and Remark (10.15) of 
Ref.\ (\cite{SWZ}) actions with Rokhlin dimension $0$ are generic in the set of actions of $\KZ$ on $\KD$
in the sense of Def.\ (10.1) and Remark (10.2) of Ref.\ \cite{SWZ}.
\end{enumerate}
\end{proof}

\subsection{Topological T-duality for $\KD$ with purely infinite fibers \label{S2TDPurInf}}
In this section we prove that there are unique lifts of the circle action on $X$ to
$\KR-$actions on $\KD$ for fiber algebras $\A$ which are purely infinite,
that is $\A$ is isomorphic to one of the two Cuntz Algebras $\KO_2$ or $\KO_{\infty}$ or
$\A$ is isomorphic to $M_{p^{\infty}} \otimes \KO_{\infty}-$a tensor
product of an infinite UHF-algebra with the infinite Cuntz algebra .
This extends the formalism of Topological T-duality for continuous-trace $C^{\ast}-$algebras
to these $C^{\ast}-$algebras.

We would like to define a theory of Topological T-duality similar to the theory of Topological
T-duality for continuous-trace $C^{\ast}-$algebras in Ref.\ \cite{JMRCBMS}
for $C^{\ast}-$algebras like $\KD$ in the previous subsection.

We had argued in Thm.\ (\ref{ThmDPrim}) above that $\Prim(\KD) \simeq X.$ 
Thus, as in the theory of Topological T-duality for continuous-trace algebras, we may try to
lift the circle action on $X$ to a $\KR-$action on $\KD$ covering the circle action on $\Prim(\KD) \simeq X.$
The T-dual should be given by the crossed product $\KD$ by the $\KR-$action.
Note that if there is such a lift as the above, it would have a stabilizer, since the $\KR-$action
covers the circle action on $X.$ However, crossed products of $\KD$ by $\KR-$actions with stabilizers 
have not been studied directly.

Currently, it is not clear if $S^1-$actions on $\Prim(\KD)$ lift to unique Rokhlin 
$\KR-$actions on $\KD$ with stabilizers for {\em every}
self-absorbing $C^{\ast}-$algebra $\A.$ 

We point out below that for purely infinite $\A$ 
a theorem of Szabo gives a unique lift and hence a well-defined T-dual $C^{\ast}-$algebra.
In the next section, we prescribe a way to obtain a T-dual $C^{\ast}-$algebra even if we cannot
find a lift of the circle action on $\Prim(\KD).$

Let $\KD,X$ be as above and let $\A$ be purely infinite.
Then, $\A$ is isomorphic to $\KO_2, \KO_{\infty}$ or $M_{p^{\infty}} \otimes \KO_{\infty}$
and $\A$ is $\KO_{\infty}-$absorbing.
Then, Thm.\ (B) of Szabo, Ref.\ \cite{S1} shows
that if there is a Rokhlin lift of an $\KR-$action on $\Prim(\KD),$ this
lift will be unique up to cocycle conjugacy.
Thus, in this case, if we can find a Rokhlin lift,  the circle action on $\Prim(\KD)$ should lift to a 
{\em unique} Rokhlin $\KR-$action on $\KD$ \label{PurInf} up to cocycle conjugacy.

Thus, we would expect a well-defined theory of Topological T-duality to be possible
for the $C^{\ast}-$algebras $\KD$ above with purely infinite fibers.

In particular for this class of $C^{\ast}-$algebra we would always have a unique `T-duality diamond'
in Eq.\ (\ref{TDia1}) with $\Prim(\KD) \simeq X:$ 

\begin{equation}
\begin{tikzcd}
   & \C \simeq \CrPr{\KD}{\KZ}{\alpha} \arrow[dl] \arrow[dr]  &  \\
\KD \arrow[dr]  &   &  \KD^{\#} \simeq \CrPr{\KD}{\KR}{\gamma} \arrow[dl] \\
                  &  \B \simeq \CrPr{\KD}{S^1}{\beta}  & \label{TDia1}
\end{tikzcd}
\end{equation}

Taking $\Prim$ of each $C^{\ast}-$algebra in the above diagram
gives a diamond diagram of spaces (see Eq.\ (\ref{TDia2}) ) 
as in ordinary $C^{\ast}-$algebraic T-duality (\cite{JMRCBMS}).

\begin{equation}
\begin{tikzcd}
   & \Prim(\C) \arrow[dl] \arrow[dr] &  \\
\Prim(\KD) \arrow[dr]  &   &  \Prim(\KD^{\#}) \arrow[dl] \\
                  & \Prim(\B) & \label{TDia2}
\end{tikzcd}
\end{equation}

It is not clear that $\KD^{\#}$ is a section algebra of a locally trivial bundle of the form $\A \otimes \Cpct$
with $\A$ strongly self-absorbing. Due to this it is not clear that $\Prim(\KD^{\#})$ is a nontrivial topological
space. However, we may always view it as a noncommutative space.
It is clear that a principal circle bundle $\pi:X \to Y$ T-dualizes to a space with a circle action $\pi^{\#}:X^{\#} \to Y^{\#}.$
Locally by the Mackey machine it is clear that the T-dual should be of the form
$U \times S^1$ where $U \subseteq  Y^{\#}.$ 
However, this might not be true globally and there might not be a T-dual
topological space. 

By the Connes-Thom isomorphism, we will obtain a degree-reversing isomorphism in $K-$theory
between the $K-$theory of $\KD$ and the $K-$theory of the crossed product
$C^{\ast}-$algebra $\CrPr{\KD}{\KR}{\gamma}.$ 

Physically, this can be interpreted as a degree-reversing isomorphism between the twisted $K-$theory of $X$ and the $K-$theory of Topological T-dual $C^{\ast}-$algebra 
$\KD^{\#}.$ If $\KD^{\#}$ is of the form of section algebra of a locally 
trivial bundle of $C^{\ast}-$algebras of the form $\A \otimes \Cpct,$ with 
$\A$ strongly self-absorbing above, then we would obtain an isomorphism between the corresponding twisted $K-$theories on either side.
In this paper, we interpret this isomorphism as corresponding to the mapping of $D-$brane charge on
both sides of the duality by the T-duality transformation as in the case of continuous-trace $C^{\ast}$-algebras
(see Refs.\ \cite{MRCMP, JMRCBMS} for details).

Further, as described in Ref.\ \cite{Brodzki2}, it is possible to define a sensible notion of $D-$brane charge
for $D-$branes on spaces which are 'noncommutative manifolds'. In particular, the $D-$brane
charge was found to be given by an explicit formula Eq.\ (6.20) of Ref.\ \cite{Brodzki2} which is
a generalization of the Minasian-Moore formula for $D-$brane charge to noncommutative spacetimes.

The above class of $C^{\ast}-$bundles
together with their crossed products by $\KR-$actions satisfy the condition for being 
noncommutative manifolds and the above Topological T-duality theory for
the $C^{\ast}-$algebras $\KD$ with purely infinite fibers should be an example of an axiomatic T-duality
in the sense of Ref.\ \cite{Brodzki1} provided one could define the $KK-$equivalence required by that definition
for these $C^{\ast}-$algebras $\KD.$ It is interesting to conjecture that the formula for the charge of $D-$branes on spaces described by $\KD$ above should be given by a formula like Eq.\ (6.20) of Ref.\ \cite{Brodzki2}.
Physically for $\A \simeq \KO_2$ this would correspond to a 'noncommutative Minasian-Moore formula' 
for $D-$branes in non-anticommutative spacetimes with $H-$flux.

Unfortunately the only other physical example should be with fiber $\A \simeq M_{2^{\infty}}$ and for these fiber algebras
$\KD$ is stably finite and, as we will see below in Subsec. (\ref{S2TDStaFin})
the crossed product construction will not work for these $C^{\ast}-$algebras.

Note that Mahanta (Ref.\ \cite{Mah1}) has defined a theory of Topological T-duality
for $\KO_{\infty}-$absorbing $C^{\ast}-$algebras using the theory of noncommutative motives and
the $KK-$theory definition of Topological T-duality. It is interesting to ask whether this theory is
related to the theory given above?

\subsection{Topological T-dual for stably finite fibers\label{S2TDStaFin}}
In this section we assume that $\KD$ has stably finite fibers, i.e. $\A$ is of the form
$M_{p^{\infty}}$ a UHF-algebra of infinite type or $\JZ$ the Jiang-Su algebra. Our remarks should
also apply to the case when $\A$ is the Razak-Jacelon algebra $\KW.$

In the case of stably finite fibers we distinguish two cases depending on the characteristic class
of $\KD-$ torsion and non-torsion characteristic class:
\begin{enumerate}
\item{{ $\delta(\KD)$ torsion:}} When 
$\A$ is stably finite, $X$ is a finite $CW-$complex and $\KD$ has torsion characteristic class in $\Bun_X(\A \otimes \Cpct),$
we will show (see the end of Subsec.\ (\ref{TorsTD}) below) that we may calculate the $T-$dual easily for a restricted class of $\KR-$actions. 

\item{{ $\delta(\KD)$ non-torsion:}}When $\A$ is stably finite, there is no way, at present,
of guaranteeing that there is a lift or that a given lift is unique. 
As a there is no way to pick a unique lift of the circle action on $X$
to a $\KR-$action, there is no way to guarantee that the T-dual is unique. 

However the procedure described in Subsec.\ (\ref{S2TDStaFin}) below lets one pick a unique $\KR-$action and a unique T-dual for {\bf stably finite}
$C^{\ast}-$algebras possibly with some extra data. 
\end{enumerate}

Extending Topological T-duality to locally trivial bundles of purely infinite strongly 
self-absorbing $C^{\ast}-$algebras {\em without arbitrary choices} is possible due to
Szabo's theorem as discussed above.
 
I propose that Topological T-duality be extended to locally trivial bundles of 
stably finite strongly self-absorbing $C^{\ast}-$algebras
by the following method-which I argue gives a unique, possibly noncommutative T-dual. 

As in the previous section, let $p:X \to B$ be a principal circle bundle. Let $\KD \in \Bun_X(\A \otimes \Cpct)$ where 
$\A$ is a strongly self-absorbing $C^{\ast}-$algebra.

We use the following procedure to find the T-dual when we cannot find an $\KR-$action 
covering the circle action on $X.$
\begin{enumerate}
\item We suppose given a spectrum-preserving Rokhlin automorphism $\alpha$ of $\KD$
which commutes with the translation action of the circle. We pick one automorphism in
a given equivalence class (usually cocycle conjugacy), 
perhaps by fixing some datum describing the equivalence class. For example, 
in some cases, for tracial Rokhlin actions (see Ref.\ \cite{EEK1} for example), the values
of the trace describe the collection of cocycle conjugacy classes.
\item We calculate $\B \simeq \CrPr{\KD}{\KZ}{\alpha}$ and take $\B$ as the
(possibly noncommutative) $C^{\ast}-$algebra associated to the correspondence space.
\item The circle action on $X$ gives rise to a circle action $\beta$ on $\B.$
\item We take the 'T-dual' as the $C^{\ast}-$algebra $\CrPr{\B}{S^1}{\beta}$ which is the
'quotient of $\B$ by $\beta$' 
\end{enumerate}

The above argument covers all the known stably finite strongly self-absorbing $C^{\ast}-$algebras. 
Thus, we would expect the above procedure to always give a unique T-dual once we fix the equivalence
class of the original actions. For example, in certain cases, for actions possessing the tracial Rokhlin property 
the trace of the automorphism would parametrize such equivalence classes. In such cases
we attach the value of the trace of the automorphism as a parameter in the duality. We do not know
the physical significance of this at present.

In the above procedure, it is clear that in the above we are calculating the crossed product by a specific $\KR-$action,
namely the $\KR-$action on $\Ind_{\KZ}^{\KR}(\A \otimes \Cpct).$
It is not clear at present that this action is a Rokhlin $\KR-$action or that there is a unique $\KR-$action covering the circle action on $X$  in every case.
Indeed, we will prove below in a specific example (\ref{ExM2InfS1}) that such an action {\em cannot} be Rokhlin.

{\em We propose that the above procedure be used to define the Topological T-dual when a unique lift of
the circle action on $X$ to a $\KR-$action on $\KD \in \Bun_X(\A \otimes \Cpct)$ is not available.} This
'Topological T-dual' will be unique up to varying the cocycle conjugacy class of the $\KZ-$automorphism and the cocycle conjugacy
class of the circle action commuting with it by the above argument. 

We now show that this 'Topological T-dual' is
actually a crossed product $C^{\ast}-$algebra of a known type:
Let $p:X \to B$ be a principal circle bundle and let $\KD \in \Bun_X( \A \otimes \Cpct).$
Let $\alpha$ be a Rokhlin spectrum-fixing $\KZ-$action on $\KD.$
Let $\delta$ be a lift of the circle action on $X$ to a (not necessarily Rokhlin) circle action on $\KD.$ 

Let ${\Hil}_{\alpha}$ be $\KD$ with the usual Hilbert $C^{\ast}-$module structure given by twisting
by $\alpha.$ Let $\KO_{\KD}({\Hil}_{\alpha})$ be the Cuntz-Pimsner algebra associated to $\KD$ by the Hilbert $C^{\ast}-$module
${\Hil}_{\alpha}.$
 
Note that the circle action $\delta$ lift to a circle action $\zeta$ on ${\Hil}_{\alpha}$ 
(see Ref.\ \cite{Deaconu} and references therein)
By Ref.\ \cite{Deaconu}, $\zeta$ gives a natural circle action $\gamma$ on the Cuntz-Pimsner algebra $\KO_{\KD}({\Hil}_{\alpha}).$

We can now calculate the T-dual we would obtain from $\KD$ if we were to 
follow the procedure outlined in Subsec.\ (\ref{S2TDStaFin}). It is strange that we can
calculate the T-dual explicitly, but it turns out to be a Cuntz-Pimsner algebra naturally
associated to $\KD.$ 
\begin{theorem}
\leavevmode
\begin{enumerate}
\item The T-dual of $\KD$ using the procedure of Subsec.\ (\ref{S2TDStaFin}) is a crossed product of $\KO_{\KD}({\Hil}_{\alpha})$ by the
circle action $\gamma$ and is isomorphic to
\begin{gather}
\CrPr{\KO_{\KD}({\Hil}_{\alpha})}{S^1}{\gamma} \simeq  \KO_{\CrPr{\KD}{S^1}{\delta}}(\CrPr{{\Hil}_{\alpha}}{S^1}{\zeta}) 
\end{gather}
This T-dual does not depend on the cocycle conjugacy class of the $\KZ-$action $\alpha$ on $\KD.$ 
\item The crossed product of $\KD$ by a $\KR-$action $\beta$ lifting the circle action on $X$ is given by
the result in the previous part if $\alpha \simeq \beta|_{\KZ}$ and $\delta$ is trivial.
\end{enumerate}
\label{TDAll}
\end{theorem}
\begin{proof}
\begin{enumerate}
\item By a result of Pimsner (\cite{Pimsner}), the crossed product $\CrPr{\KD}{\KZ}{\alpha} \simeq \KO_{\KD}({\Hil}_{\alpha}).$
Hence, the construction of Subsec.\ (\ref{S2TDStaFin}) gives the T-duality diamond in Eq.\ (\ref{TDiaOA}) below
\begin{equation}
\begin{tikzcd}
&& \C \simeq  \CrPr{\KD}{\KZ}{\alpha} \simeq \KO_{\KD}({\Hil}_{\alpha})  \arrow[dl] \arrow[dr] & \\ 
{\Hil}_{\alpha} \arrow[r, dashed]&\KD \arrow[dr] & & {\KD^{\#}} \simeq \CrPr{\KO_{\KD}({\Hil}_{\alpha})}{S^1}{\gamma}
\simeq \KO_{\CrPr{\KD}{S^1}{\delta}}(\CrPr{{\Hil}_{\alpha}}{S^1}{\zeta})
 \arrow[dl] \\
& & \B \simeq \CrPr{\KD}{S^1}{\delta} & \\
& & \CrPr{ {\Hil}_{\alpha}}{S^1}{\zeta} \arrow[u,dashed] & \label{TDiaOA}
\end{tikzcd}
\end{equation}

Since $S^1$ is amenable, by the Hao-Ng Theorem
(see  Ref.\ \cite{Deaconu} page 6 and Ref.\ \cite{HaoNg} Thm.\ (2.10)~)
\begin{gather}
\CrPr{\KO_{\KD}({\Hil}_{\alpha})}{S^1}{\gamma} \simeq \KO_{\CrPr{\KD}{S^1}{\delta}}(\CrPr{{\Hil}_{\alpha}}{S^1}{\zeta})
\end{gather}

Note that the T-dual by the above procedure can only depend on
the trace conjugacy class of $\alpha$ and on 
the cocycle conjugacy class of the circle action $\delta$ on $\KD.$

\item This follows from the previous part since the $\KR-$action would give the T-duality diamond-like Diagram (\ref{TDiaOA}) above
because, as discussed in Sec.\ (\ref{S2TDStaFin}) the maps in the above diagram are just the maps induced by induction in 
stages \cite{DWill} for a $\KR-$action on $\KD.$ 
Due to this, $\beta$ must restrict to $\alpha$ on $\KZ \into \KR.$ 
We can view the circle action $\delta$ as a $\KR-$action $\tilde{\delta}$ on $\KD$ by composing 
$\KR \to \KR/\KZ$ with $\delta:\KR/\KZ \to \Aut(\KD).$ 
Then, $\tilde{\delta} \circ \beta^{-1}$ is a spectrum-fixing $\KR-$action on $\KD.$

\end{enumerate}
\end{proof} 

In the above Theorem, the circle action $\delta$ need not be Rokhlin, however
Prop.\ (5.6) of Ref.\ \cite{S2}
shows that the $\KR-$action $\tilde{\delta}$ is Rokhlin if the circle action $\delta$ is Rokhlin.
%

Note that the above construction is reversible and taking the crossed product of $\KD^{\#}$
by the correct dual actions will return the original $C^{\ast}-$algebra $\KD.$
Thus the construction is involutory as a duality should be.

{\flushleft {\bf Note:}} Let $X$ be a point then, then $\KD \simeq \A \otimes \Cpct$ with $\A$ stably finite. 
We would like to T-dualize these $C^{\ast}-$algebras using the method outlined in this Subsection.
For this, we note the following results about automorphisms of stabilized stably finite strongly self-absorbing $C^{\ast}-$algebras
\begin{itemize}
\item{$\A \otimes \Cpct \simeq \KW \otimes \Cpct$:} In Ref.\ \cite{N1}, Nawata has defined a Rokhlin property for automorphisms
of the stabilized Razak-Jacelon algebra. This algebra posesses a unique trace $\tau.$ The author terms an automorphism $\alpha$ 
to be trace scaling if there exists an invariant-the trace scaling factor $\lambda(\alpha)-$ 
such that $\lambda(\alpha) \tau(a) = \tau(\alpha(a)), \forall a \in \KW \otimes \Cpct.$
He shows that two trace scaling Rokhlin automorphisms $\alpha$ and $\beta$ (i.e. two Rokhlin $\KZ-$actions)  on
$\KW \otimes \Cpct$ are outer conjugate to each other if and only if $\lambda(\alpha) = \lambda(\beta).$ 

\item{$\A \otimes \Cpct \simeq M_{p^{\infty}} \otimes \Cpct$:} In Ref.\ \cite{EEK1}, Elliott, Evans and Kishimoto classified trace scaling
automorphisms of stabilized UHF-algebras and showed that all automorphisms $\alpha$ of a stabilized UHF algebra with the same trace scaling factor $s(\alpha) \neq 1$ are outer conjugate. They also determined the structure of the crossed product of $M_{r^{\infty}} \otimes \Cpct$ by $\alpha$
for some supernatural number $r^{\infty}.$ 

\item In Ref.\ \cite{Sato1}, Sato has defined a weak Rokhlin property for automorphisms of the Jiang-Su algebra
and shows that $\KZ-$actions on $\JZ$ with this property are all outer conjugate.  \label{StaFin}
\end{itemize}

If two automorphisms are outer conjugate, the crossed product by these Rokhlin $\KZ-$actions are isomorphic.
For each of the above classes of {\em stably finite} strongly self-absorbing algebras, it follows that all the outer
conjugacy classes of Rokhlin $\KZ-$actions on $\A \otimes \Cpct$ with $\A$ stably finite (except $\A \simeq \JZ$)
are classified by the trace. Thus, the isomorphism classes of the Topological T-dual of $\A \otimes \Cpct$ are classified by the
trace as well.

Simple $C^{\ast}-$algebras like the Razak-Jacelon and UHF $C^{\ast}-$algebras above as correspond to
`noncommutative points'.  Also, it is well known that the Cuntz algebra $\KO_2$ should be viewed as a `noncommutative circle'.
We will discuss the T-dual of an element of $\KD \in \Bun_{S^1}(M_{2^{\infty}} \otimes \Cpct)$ in Sec.\ \ref{SecMathEx} below using
the discussion above. We know that $\Prim(\KD) \simeq S^1.$ 
We will see that in some cases, the Topological T-dual is a stabilized Cuntz algebra $\KO_2 \otimes \Cpct$ which we view
as a noncommutative circle.

\subsection{Maps in $K-$theory \label{S2KTM}}
In this section we obtain natural maps between the $K-$theory of $\KD$ and the $K-$theory of the proposed 
Topological T-dual $\KD^{\#}$ for $\KD$ above with $\A$ purely infinite or stably finite. 
In this Section, we assume that $X, \A, \KD, \KD^{\#}, \alpha$ and $\gamma$ are as in Thm.\ \ref{TDAll} above.

For Topological T-duality for $\KD$ with $\A$ purely infinite, we can lift the circle action on $X$ to a $\KR-$action on $\KD$ and T-dualize in the
sense of Mathai and Rosenberg. The T-dual of $(\KD, \alpha)$ will be $(\KD^{\#}, \alpha^{\#})$ and 
$\KD^{\#} \simeq \CrPr{\KD}{\KR}{\alpha}.$ 
As in Topological T-duality for continuous-trace $C^{\ast}-$algebras, \cite{JMRCBMS} 
there will be a degree-reversing isomorphism  
$K_{\bullet}(\KD) \to K_{\bullet+1}(\KD^{\#})$
induced by the crossed product. 

If $\KD^{\#}$ is also a section algebra of a locally trivial bundle with fiber a stabilized strongly self-absorbing
$C^{\ast}-$algebra, then the two $K-$theory groups on either side of the isomorphism correspond to
higher twisted $K-$theory of $\Prim(\KD)$ and $\Prim(\KD^{\#})$ respectively. In this case,
we interpret the above isomorphism as an isomorphism on higher twisted $K-$theory of the two sides of the duality
induced by Topological T-duality \cite{MRCMP} for $\A$ purely infinite.

Is there an isomorphism in $K-$theory for $\KD$ with stably finite fibers? 
We now argue that for stably finite fibers there is natural map between $K-$theory of the
$C^{\ast}-$algebras $\KD, \KD^{\#}$ above. 

The $K-$theory of the crossed product $\CrPr{\KO_{\KD}(\Hil_{\alpha})}{S^1}{\gamma}$ which we suggested was the
T-dual of $\KD$ with stably finite fibers in Thm.\ (\ref{TDAll})
above may be calculated using Theorem (A) of Ref.\ \cite{Schafhauser} and depends on
the morphism $[\KH_{\alpha}]$ induced by $\Hil_{\alpha}$ on $K_{\ast}(\KD).$
\begin{theorem}
Assume $\A, \KD, X$ and all $\KZ-$, $S^1-$ and $\KR-$actions on them are as in Thm.\ (\ref{TDAll})
\leavevmode
\begin{enumerate}
\item There is a natural map $\phi$ between the $K-$theory of $\KD$ and the $K-$theory of the 'T-dual' $\KD^{\#}.$
\item In case the T-dual is given by crossed product by a Rokhlin $\KR-$action lifting the
circle action on $X,$ $\phi$ is an isomorphism.
\end{enumerate}
\label{ThmIsoDDual}
\end{theorem}
\begin{proof}
\leavevmode
\begin{enumerate}
\item By Theorem (A) of Ref.\ \cite{Schafhauser}, there is a natural isomorphism
$$
\varinjlim (K_{\ast}(\KD), [\KH_{\alpha}]) \simeq K_{\ast}(\CrPr{\KO_{\KD}(\Hil_{\alpha})}{S^1}{\gamma}) 
$$
which is an isomorphism on the $K-$groups which 
intertwines the natural automorphism $[\hat{\gamma}]$ on the $K-$theory of the crossed product by the natural dual action $\hat{\gamma}$ on the crossed product with the automorphism $[\KH_{\alpha}]$ on $K_{\ast}(\KD).$

We can compose this with the natural morphism $K_{\ast}(\KD) \to \varinjlim (K_{\ast}(\KD), [\KH_{\alpha}])$ to
get the required map 
$
\phi: K_{\ast}(\KD) \to  K_{\ast}(\CrPr{\KO_{\KD}(\Hil_{\alpha})}{S^1}{\gamma}) 
$
\item It is not clear whether the map $\phi$ in the previous item is an isomorphism. 
However, since the construction in Thm.\ (\ref{TDAll}) reduces to the
crossed product by induction by stages when the original $S^1-$action on
$\Prim(\KD)$ lifts to a Rokhlin action this map will be a degree reversing isomorphism 
in this case due to the Connes-Thom Isomorphism.
\end{enumerate}
\end{proof}

The $K-$theory of the $C^{\ast}-$algebra $\KD$ corresponds 
-see discussion in item (3) of Subsection\ \ref{SubSecNCGFluxB} below-to a higher twisted $K-$theory of the space 
$X = \Prim(\KD).$ We argue in that Subsection that the higher twisted $K-$theory associated
to $\KD$ above is the $D-$brane charge associated to $D-$branes propagating on
the background $X$ with background fluxes present.

When $\KD^{\#}$ is also a section algebra of a $C^{\ast}-$bundle with
fiber a strongly self-absorbing $C^{\ast}-$algebra then we interpret
$K_{\bullet}(\KD^{\#})$ as the higher twisted $K-$theory as the 
$D-$brane charge associated to $D-$branes propagating on the background
$X^{\#}$ with dual background fluxes present.

We interpret the above isomorphism in $K-$theory (for purely infinite fibers) 
and the above map in $K-$theory (for stably finite fibers) as a mapping of
$D-$brane charges from the original theory to the dual under a 'T-duality like' 
transformation.

\subsection{Generalizations \label{S2Gen}}
The above formalism defines a T-dual $C^{\ast}-$algebra for any $C^{\ast}-$ bundle $\KD$ of the type
defined above. It is interesting to relate this T-dual to the {\em Axiomatic T-duality} formalism of Ref.\ \cite{Brodzki1} Sec.\ (5.3).

If the above extension of T-duality was an example of Axiomatic T-duality then,
the assignment of the T-dual $C^{\ast}-$algebra $T(\KD)\simeq \CrPr{\KO_{\KD}({\Hil}_{\alpha})}{S^1}{\gamma}$ using the above procedure should give a covariant functor $\KD \mapsto T(\KD):$ 
\begin{enumerate}
\item If the circle action on $X$ lifts to a Rokhlin $\KR-$action on $\KD,$ it is clear that the crossed product by a 
Rokhlin action of $\KR$ gives a covariant functor on the category of
locally trivial $C^{\ast}-$bundles $\KD$ with fiber of the form $\A \otimes \Cpct$ above together with
Rokhlin $\KR-$actions and $\KR-$equivariant morphisms to the diagram category of diagrams of $C^{\ast}-$algebras
of the form in Fig.\ \ref{TDiaOA}. This will be the case for $\A$ purely infinite.

\item If a Rokhlin lift does not exist as above, the procedure given above in Thm.\ (\ref{TDAll}) 
is also a covariant functor with the same target as in the previous item, provided one considers the source category the category of locally trivial
$C^{\ast}-$bundles $\KD$ with fibers of the form $\A \otimes \Cpct$ above with Rokhlin $\KZ-$actions and $\KZ-$equivariant morphisms-this would cause the Hilbert $C^{\ast}-$module
$\Hil_{\alpha}$ and the full diamond diagram in Thm.\ (\ref{TDAll}) to map functorially. {\em Since Rokhlin $\KZ-$actions should exist 
when Rokhlin $\KR-$actions don't -by the argument in Subsec.\ (\ref{S2TDStaFin}) above- one of these two should always be possible for this class of $C^{\ast}-$algebras.}
\end{enumerate}
However, it is not clear if there would be a natural element $\alpha_{\D}$ in $\KK_{\bullet}(\D,T(\D))$ which is a 
$KK-$equivalence given by the above procedure possessing the properties required by the definition of Axiomatic T-duality.
However, we still have the map $\phi$ between the $K-$theory of $\KD$ and $T(\KD)$ given by Thm. (\ref{ThmIsoDDual}) above. 

Also, it is interesting to point out in the above, that the Hilbert $C^{\ast}-$module need not be the module $\Hil_{\alpha}$ used above.
It is clear that the above T-duality diamond diagram will remain but, it is not clear what relation it would have to String Theory.
It would be interesting to see if Hilbert $C^{\ast}-$modules apart from 
$\Hil_{\alpha}$ above could describe some known phenomenon in string theory.

\section{Mathematical Examples \label{SecMathEx}}
\subsection{Introduction}
Topological T-duality is an attempt to make a model of the T-duality symmetry of Type II
string theory using techniques from noncommutative geometry and algebraic topology.

Table (\ref{TabTD}) summarizes our conjectures for the physical interpretation of $\Bun_X(\A \otimes \Cpct)$
above for varying $\A.$ Note the case $\A \simeq \KC$ corresponds to Topological T-duality
for the Type II string theory with supersymmetry ignored. This table is based
on the calculations in the following sections.

\begin{center}
\begin{table}[h!b]
{\scriptsize
\caption{ Fiber algebras and associated T-duality}
\label{TabTD}
\begin{tabular}{|c|c|c|}
\hline
\makecell{\bf Fiber algebra $(\A) $ \\ \bf Actual Fiber is $(\A \otimes \Cpct)$} & \makecell{\bf Type of T-duality} & \makecell{\bf String Theory \\ \bf associated to T-duality }\\
\hline \hline
$\KC$ & \makecell{\bf Circle or Torus action\\ Non-zero $H-$flux \\ Supersymmetry ignored} & Type II A and Type IIB \\
\hline 
\makecell{\bf Stably Finite \\ \bf Fiber:}  &   &  \\
\hline 
$M_{p^{\infty}},$    $p=2$& \makecell{ Supersymmetric Background,\\ $RR-$flux or \\ $H-$flux present \\ {\bf Fermionic T-duality}} & \makecell{$R_{fiber} \gg l_s$ \\ Type II and \\ Type II* theories} \\
\hline
$\JZ$ & \makecell{String Theories with\\ Compactified Timelike directions \\ i.e. Closed Timelike Loops \\ {\bf Timelike T-duality}} & \makecell{Background with \\ nonstandard signature 
\\$(m,n)$ with $m > 1$} \\
\hline
$\KW$ & Unknown at present &  \\
\hline
{\bf Purely Infinite Fiber:} & &\\
\hline
$\KO_2$ & \makecell{Non-anti-commutative \\ Superspace \\ $N=1/2$ SUSY \\ {\bf Fermionic T-duality}} & \makecell{\tiny $ \theta_a \theta_b  + \theta_b \theta_a = \epsilon_{ab}$ \\
Type II theory\\ on NAC superspace } \\
\hline
$\KO_{\infty}$ & None at present&\makecell{$R_{fiber} \ll l_s$ \\ Type IIA, Type IIB string theories}\\
\hline
$M_{p^{\infty}} \otimes \KO_{\infty}$ &None at present & None at present\\
\hline
\end{tabular}
}
\end{table}
\end{center}

\subsection{T-duality with $\KD$ of torsion class}
\label{TorsTD}
We now consider an interesting class of examples for which the Topological T-dual
can be calculated exactly. 

We study T-duality for $C^{\ast}-$algebras $\KD$ in $\Bun_X(\A \otimes \Cpct)$ with $\A$
strongly self-absorbing and $X$ a finite $CW-$complex. In addition, we require that the characteristic
class of $\KD,$ $\delta(\KD)$ -defined in Thm. (4.5) of Ref.\ \cite{DP1}- be torsion, 
that is, the characteristic class lie in $\Tor(\bar{E}^1_D(X))$
-see Ref.\ \cite{DP2}, Thm. (2.8) and Thm. (2.10). 

By Ref.\ (\cite{DP2}), this is equivalent
to requiring that each of the rational characteristic classes $\delta_k(\KD)$ vanish and
also requiring that there exist $n >0$  such that $\KD^{\otimes n} = C(X) \otimes \A \otimes \Cpct.$

In addition, by Thm.\ (2.8)(iii) in Ref.\ \cite{DP2}, $\KD$ is isomorphic to the stabilization
of a unital, locally trivial continuous field of $C^{\ast}-$algebras $\B$ over $X$ with fiber $M_m(A)$ for some
$m.$ The $C^{\ast}-$algebra $\B$ is not unique but only unique up to the equivalence given in
Thm.\ (2.9) of Ref.\ \cite{DP2}.

We consider a spacetime whose underlying topological space is homeomorphic to $\KR^n \times X$ 
with $X$ a finite $CW-$complex.
Let $X$ be a principal circle bundle over a base space $W \simeq X/S^1.$
Suppose we were trying to calculate the T-dual of $X$ with background fluxes whose characteristic
classes were torsion. Then, by the above, we should consider $C^{\ast}-$algebras $\KD$ in
$\Bun_X(\A \otimes \Cpct)$ with torsion characteristic class.
Note that $\KD \simeq \B \otimes \Cpct$ so $\Prim(\KD)$ and $\Prim(\B)$ 
are homeomorphic. Also, $K-$theory doesn't change on stabilization, so
$K_{\ast}(\KD) \simeq K_{\ast}(\B).$ 

We may thus get some idea about the T-dual of $\KD$
by considering crossed products of $\B$ by suitable $\KR-$actions and stabilizing them.
This gives us a large number of examples since $\B$ is {\em unital} unlike $\KD$.
By the above the T-dual $C^{\ast}-$algebra 
should be given by the crossed product by the lift of the circle action on $X$ to a $\KR-$action when
$\A$ is purely infinite. We now concentrate on the case when $\A$ is stably finite.

For Rokhlin $\KZ-$actions on $\B,$ we use Cor.\ (5.14) and Cor.\ (5.16) of Ref.\ \cite{Phil}.

We now show that Rokhlin circle actions on $\B$ induce Rokhlin circle actions on $\KD.$
\begin{lemma}
\leavevmode
Let $X$ be a finite $CW-$complex.
Suppose $\B,\B_1,\B_2$ are section algebras of locally trivial $C^{\ast}-$bundles over $X$ with
fibers $M_n(\A), M_{n_i}(\A), i=1,2$  such that $\B \otimes \Cpct \simeq \KD$ with $\KD$ of
torsion characteristic class and similarly for $\B_1,\B_2.$ 
\label{LemTorsionDCirc}
\begin{enumerate}
\item Every Rokhlin action of a circle on $\B_1$ induces a Rokhlin action of the circle on $\KD_1.$
\item If $\B \simeq \B_1 \otimes \B_2$ and $\alpha_1, \alpha_2$ are Rokhlin circle actions on $\B_1$ 
and $\B_2$ respectively, then $\alpha_1 \otimes \alpha_2$ is a Rokhlin circle action on $\B.$
\item If $\A \otimes \KO_2 \simeq \KO_2,$ then the set of Rokhlin circle actions on $\B$ are a
dense $G_\delta-$set of the space of circle actions on $\B.$
\end{enumerate}
\end{lemma}
\begin{proof}
\leavevmode
By assumption $\KD, \KD_i \in \Bun_X(\A \otimes \Cpct)$
and the characteristic classes of $\KD,\KD_i, i=1,2$ are torsion.
\begin{enumerate}
\item By the above, $\KD \simeq \B \otimes \Cpct.$ Also, $\B$ is $\sigma-$unital since it is unital.
Let $\alpha$ be a Rokhlin action of the circle on $\B.$ 
Let $\xi \mapsto \gamma_\xi$ be the translation action of $S^1$ on $L^2(S^1)$ 
It can be proved that the translation action lifts to a continuous map $\lambda: S^1 \to \U(L^2(S^1))$
such that $\gamma_\xi(a) = \Ad(\lambda(\xi))(a), a \in L^2(S^1).$
(for example, see the book by Raeburn and Williams Ref.\ \cite{RaeWill} ).
Also it is easy to see that $\alpha \otimes \Ad(\lambda)$ is an action of the circle on $\A \otimes \Cpct$
The action $\alpha \otimes \Ad(\lambda)$ is Rokhlin by Remark (2.8) of Ref.\ \cite{G3}-
also see remarks before Def. (2.2) in the same paper. 
\item This follows from Prop.\ (2.5) of Ref.\ \cite{G0}, since $\B_i, i=1,2,3$ are unital.
\item Since $\B$ is separable and unital, we use the definition Def.\ (2.14) for the space of circle actions on $\B.$
Since the circle is compact, by Cor.\ (2.22) (i) of Ref.\cite{G2}, the set of Rokhlin circle actions on $\B$ form a
dense $G_{\delta}-$set of the space of circle actions on $\A \otimes \Cpct.$
\end{enumerate}
\end{proof}

We can also use the unitality of $\B$ to find enough $\KZ-$actions on $\B.$
\begin{lemma}
Let $X$ be a finite $CW-$complex. Let $\B$ be the section algebra of a $C^{\ast}-$bundle
over $X$ with fiber $M_n(\A)$ with $\A$ a strongly self-absorbing $C^{\ast}-$algebra.
\label{LemTorsionDZ}
\leavevmode
\begin{enumerate}
\item For every $\B$ above,  automorphisms of $\B$ of Rokhlin dimension $\leq 1$ are generic.
\item If $\A$ is UHF-stable for a UHF-algebra of infinite type, then Rokhlin automorphisms of $\B$ are generic.
\end{enumerate}
\end{lemma}
\begin{proof}
\leavevmode
\begin{enumerate}
\item Since $\A$ is strongly self-absorbing, $\A$ is always $\JZ-$absorbing, hence, so is $\B.$ 
Also, by the above, $\B$ is unital, separable and $\JZ-$absorbing.
By Thm.\ (3.4) of Ref.\ \cite{SWZ}, automorphisms of $\B$ of Rokhlin dimension $\leq 1$ are generic in the set of
all automorphisms of $\B$ in the sense of Ref.\ \cite{SWZ}.
\item Suppose $\A$ was UHF-stable for a UHF-algebra of infinite type, then so is $\B$. 
By the previous part, $\B$ is unital, separable and $\JZ-$absorbing as well.
Hence, by Thm.\ \cite{SWZ}, Rokhlin automorphisms are generic in the set of all automorphisms
on $\B$ in the sense of Ref.\ \cite{SWZ}. 
\end{enumerate}
\end{proof}

We now try to find a generalized T-duality diagram in the sense of Sec.\ (\ref{TDAll}) above.
If we are given a $C^{\ast}-$ algebra $\KD$ with torsion characteristic class, we pick
a $\B$ above such that $\B \otimes \Cpct \simeq \KD.$ Note that $\Prim(\B) \simeq \Prim(\KD)$ as
discussed above. Note that $\B$ is not unique as shown in Ref.\ \cite{DP2} however, every choice
of $\B$ has the same spectrum as $\KD.$ The only thing which changes is the choice of fiber algebra
-each choice of $\B$ corresponds to a fiber algebra of the form $M_n(\A)$ with differing values of
$n.$ This should not affect anything since we are only interested in the result obtained by stabilizing
$\B.$

We restrict our choice of group actions on $\B$ to tensor product actions of the type $\alpha \otimes \Ad$
where $\alpha$ is a Rokhlin group action on $\B.$ It is always possible to pick such a $\alpha$
by Lemma (\ref{LemTorsionDCirc}) above. By the same Lemma,
this gives a natural Rokhlin circle action $\alpha \otimes \Ad(\lambda)$ on $\KD.$
With this choice of actions, we can restrict ourselves to examining T-duality diagrams of the
form in Eq.\ (\ref{TDiaTors}) below.

We pick a Rokhlin $\KZ-$action $\gamma$ on $\B$ or a Rokhlin $\KZ-$action of dimension at most
by Lemma (\ref{LemTorsionDZ}). We pick the Rokhlin $\KZ-$action on $\B$ such that it commutes
with $\alpha$ above. Since Rokhlin $\KZ-$actions of dimension at most one are dense in the set of
all Rokhlin actions on $\B,$ this should be possible. 
This induces a natural $\KZ-$action $\theta^{\#}$ on $\B^{\alpha}$ 
which commutes with $\theta.$

Let $\B^{\alpha}$ be the set of elements of $\B$ fixed by $\alpha.$
By Thm.\ (2.3) of Ref.\ \cite{G4}, since $\B$ is unital,
there is an automorphism $\theta \in \B^{\alpha}$ such that
$\B$ is isomorphic to $\CrPr{\B^{\alpha}}{\KZ}{\theta}.$
Let $\theta'$ be any other automorphism of $\B^{\alpha}$ such that $(\CrPr{\B^{\alpha}}{\KZ}{\theta'}, \hat{\theta'})$
is conjugate to $(\KD, \alpha),$ then by Thm.\ (2.3)(2) of Ref.\ \cite{G4}, 
$\theta$ is unitarily equivalent to $\theta'$ by a unitary in $\B^{\alpha}.$
Also, by Cor.\ (2.5) of Ref.\ \cite{G4}, $\CrPr{\B}{S^1}{\alpha} \simeq \B^{\alpha} \otimes \Cpct.$

Thus the Topological T-duality diagram in Thm.\ (\ref{TDAll}) is now modified due to the above.
In particular, $\KD$ is now isomorphic to $\B \otimes \Cpct$ and, by the above, isomorphic to $(\CrPr{\B^{\alpha}}{\KZ}{\theta'}) \otimes \Cpct.$
Also, we must have that $\B^{\alpha}$ is isomorphic to the fixed point $C^{\ast}-$algebra of $\B^{\#}$ by
the dual automorphism $\alpha^{\#}$  for consistency.
The T-duality diagram now reduces to the two commuting $\KZ-$actions $\theta$ and $\theta^{\#}$ on $\B^{\alpha}.$ 
The T-duality swaps the two $\KZ-$actions.

\begin{equation}
\begin{tikzcd}
   & \C \simeq \CrPr{\B}{\KZ}{\gamma} \arrow[dl] \arrow[dr]  &  \\
(\CrPr{\B^\alpha}{\KZ}{\theta}) \simeq \B \arrow[dr]  &   &  \B^{\#}  \simeq (\CrPr{\B^\alpha}{\KZ}{\theta^{\#}}) \arrow[dl] \\
                  &  \B^{\alpha} \simeq (\CrPr{\B}{S^1}{\alpha})  & \label{TDiaTors}
\end{tikzcd}
\end{equation}

The data which characterizes the original 
noncommmutative space should be the characteristic class of $\KD \simeq \B \otimes \Cpct$ 
together with the automorphism $\theta$ of $\B^{\alpha}.$
The data which characterizes the T-dual space should be the characteristic class of 
$\KD^{\#} \simeq (\CrPr{\B^{\alpha}}{\KZ}{\theta^{\#}}) \otimes \Cpct$ 
together with the dual automorphism $\theta^{\#}$ on $\B^{\alpha}.$

Note that $\B, \B^{\#}$ are unital and $K_{\ast}(\B)$ is finitely generated since $X$ is finite.
Now, $\B^{\#} \simeq \CrPr{\B^{\alpha}}{\KZ}{\theta^{\#}}$ so $K_{\ast}(\B^{\#})$ is finitely
generated as well.

By Thm.\ (3.3)(2) of Ref.\ \cite{G4}, 
$$
K_0(\B) \simeq K_1(\B) \simeq K_0(\B^{\alpha}) \oplus K_1(\B^{\alpha})
$$
similarly
$$
K_0(\B^{\#}) \simeq K_1(\B^{\#}) \simeq K_0(\B^{\alpha}) \oplus K_1(\B^{\alpha})
$$
Thus, the $K-$theory of $\B$ is isomorphic to the $K-$theory of $\B^{\#}$ and the isomorphism 
may be taken to be degree reversing (since $K_{\ast}(\B)$ and $K_{\ast}(\B^{\#})$are isomorphic). 

Thus, the same holds for $\KD \simeq \B \otimes \Cpct$ since stabilization preserves $K-$theory.
  
\subsection{Fermionic T-duality}
\label{CStarEx}
In this subsection we concentrate on $\A \simeq M_{2^{\infty}}-$this corresponds
to a fermionic T-duality on the string theory action or to fermionic T-duality composed with bosonic T-duality transformation
of the string theory action. 
A few examples using the proposed T-duality procedure
are calculated below and compared with the crossed product by a Rokhlin 
action of $\KR$ whenever possible. Uniqueness of the Rokhlin action and its lifting properties are discussed. Since it is difficult to find Rokhlin lifts of
circle actions on $X,$ we do not have too many examples, but the examples below are interesting.

\begin{enumerate}
\item{\underline{\bf $M_{2^{\infty}}-$bundles over a Circle:}}
\label{ExM2InfS1}
\begin{itemize}
\item We consider a supersymmetric background $\KR^{(8,1)} \times S^1.$ 
We place a graviphoton flux on the background and perform a fermionic
T-duality. The T-dual should be non-anti-commutative superspace
over a circle. 

\item For this example let $X$ be $S^1$ viewed a trivial circle bundle over a point
i.e.,$X \simeq S^1 \times \{ \ast \} \to \{ \ast \}.$
Suppose we associate to $\{ \ast \}$ the $C^{\ast}-$algebra $M_{2^{\infty}}.$

Let $\KD \in \Bun_X(\A \otimes \Cpct)$ where $\A$ is the CAR algebra $M_{2^{\infty}}.$

It can be argued that $\KD$ is an induced algebra of the form 
$\Ind_{KZ}^{\KR}(M_{2^{\infty}} \otimes \Cpct)$
for some automorphism $\alpha$ of $M_{2^{\infty}} \otimes \Cpct.$

In Ref.\ \cite{DP1} Dadarlat and Pennig show that
the characteristic class of such bundles lies in a sum of cohomology
groups of the form $H^{\ast}(X,R^{\times}_{+})$ where
$R = K_0(M_{2^{\infty}} \otimes \Cpct).$

However, for $M_{2^{\infty}}$, it can be checked (see Ref.\ \cite{P1}) 
that $R^{\times}_{+}$ above is $\KZ$, so
the characteristic class of the CAR bundle
lies in $H^1(S^1,\KZ) \simeq \KZ.$

%
%
%

\item Now, pick a spectrum-fixing Rokhlin $\KZ-$action $\alpha$ on $\KD$ 
{\bf commuting with the translation action on the circle.}
Since $\alpha$ is a spectrum-fixing $\KZ-$action on $\KD,$ it
will act only on the fiber $M_{2^{\infty}} \otimes \Cpct$ of $\KD.$

\item The crossed product $\B \simeq \CrPr{\KD}{\KZ}{\alpha}$ is a 
$C(S^1)-$bundle. $\B$ is also $C^{\ast}-$bundle over $\Prim(\B)$
with fiber $\CrPr{(M_{2^{\infty }}\otimes \Cpct)}{\KZ}{\alpha}.$

\item Every trace scaling automorphism $\alpha$ of $M_{2^{\infty}} \otimes \Cpct$
possesses a trace scaling factor $s(\alpha)=p/q$ which is a positive rational
number (see \cite{EEK1}, pg. 77). Each of these automorphisms is Rokhlin 
-see Thm.\ (2.8) of Ref.\ \cite{EvKish1}. 
Also, by Ref.\ \cite{EEK1}, Thm.\ (7), automorphisms with the 
same trace scaling factor are outer conjugate.

Let $\alpha$ to be an automorphism with trace
scaling factor $2.$
By tensoring $\alpha$ with itself (and similarly for $\alpha^{-1}$) 
and using the fact that $M_{2^{\infty}} \otimes \Cpct $ is isomorphic to its tensor
product with itself, we can obtain automorphisms of trace scaling
factors $2^k$ for all $k \in \KZ.$

By Remark on pg.\ (84) of Ref.\ \cite{EEK1}, if $\alpha$ is an automorphism
of $M_{(pq)^{\infty}} \otimes \Cpct$ with trace $s(\alpha) = p/q,$
the crossed product $\CrPr{\C}{\KZ}{\alpha}$ is isomorphic to
$O_{|p-q| + 1} \otimes \Cpct.$

Of the above automorphisms, it is clear by examining the positive integer solutions to 
$pq = 2^m, |p - q| =1$ that only the conjugacy classes of Rokhlin $\KZ-$actions
$\alpha$ on $M_{2^{\infty}} \otimes \Cpct$-those with trace 
$s^{\ast}(\alpha) = 2/1$ and those with
trace $s^{\ast}(\alpha) = 1/2-$give $\KO_2 \otimes \Cpct$ as a crossed product.

\item We require the crossed product to have strongly self-absorbing fiber and by the
above, this fiber can only be $\KO_2 \otimes \Cpct.$ So, changing
the characteristic class of $\KZ-$action can only change the trace of $\alpha$
from $2$ to $1/2$ or vice versa. 
 
\item Thus, the correspondence space is a $C^{\ast}-$bundle over 
$\Prim(\KD) \simeq S^1$ with fiber $\KO_2 \otimes \Cpct.$
Due to the local triviality of $\D$ and the above argument,
this bundle has to be a locally trivial bundle of Cuntz algebras $\KO_2.$
In Ref.\ \cite{DP1} Dadarlat and Pennig show that
the characteristic class of such bundles lies in a sum of cohomology
groups of the form $H^{\ast}(X,R^{\times}_{+})$ where
$R = K_0(\KO_2 \otimes \Cpct).$
However, for $\KO_2$, it can be checked that $R$ above is trivial, so
the characteristic class of Dadarlat and Pennig is trivial. By Cor. (4.3)
and Cor. (4.6) of Ref.\ \cite{DP1} the above $C^{\ast}$-bundle
must be trivial as well. Changing the characteristic class of the orginal
$M_{2^{\infty}}-$bundle on $S^1$ won't change the characteristic class of
the $\KO_2 \otimes \Cpct-$bundle above, since it must always be trivial.

\item  The $\KZ-$action on $\KD$ is spectrum-fixing, and, as a result it acts only on
the fiber $M_{2^{\infty}}$ of $\KD.$ Since it commutes with the translation action on $S^1,$
the crossed product $\CrPr{\KD}{\KZ}{\alpha}$ should then be
$C(S^1) \otimes \KO_2 \otimes \Cpct.$

\item Thus the correspondence space is a trivial $\KO_2 \otimes \Cpct-$bundle over
a topological space $E \simeq S^1.$ The space $E$ carries a circle action which comes from the original circle
action on the original $S^1$ which only translates $S^1.$ (this can be seen by working out the definition of the crossed product, viewing
the algebra $\KD$ as an induced algebra of the form $\KT_{\alpha}(M_{2^{\infty}} \otimes \Cpct).$

\item The quotient of the $C^{\ast}-$bundle by this circle action is $\KO_2 \otimes \Cpct.$
\item Geometrically, the $C^{\ast}-$algebra $\KO_2$ may be viewed as a noncommutative circle bundle over a CAR algebra.
We could view this as the noncommutative geometry of the T-dual.
\item Thus, we claim that the fermionic T-dual of a supersymmetric theory over a circle is non-anti-commutative superspace
over a point and the T-dual circle action is the rotation of the generators of non-anti-commutative superspace by a phase.
This should be matched with the Buscher's rules for fermionic T-duality: It can be seen that the physical T-dual is a circle
with the fermionic directions `twisted' by the T-dual graviphoton flux into NAC superspace.
\end{itemize}

We may also argue by calculating the crossed product by $\KR:$
We note that by Lemma (3.1) and Lemma (3.2) of Ref.\ \cite{P1}, $\KD$ is isomorphic
to the induced $C^{\ast}-$algebra $\KT_{\alpha}(M_{2^{\infty}} \otimes \Cpct)$ defined there.
By Lemma (5.1) and Lemma (5.2) of \cite{P1},
the crossed product by the lift of the translation action on $\Prim(\KD)$ to $\KR$
is $\KO_2 \otimes \Cpct$ which agrees with the above. 

\begin{theorem}
If the T-dual by a crossed product by a $\KR-$action is $\KO_2 \otimes \Cpct,$ 
then the dual action on $\KO_2 \otimes \Cpct$ cannot be Rokhlin.
\end{theorem}
\begin{proof} 
Suppose the dual action on $\KO_2 \otimes \Cpct$ was
Rokhlin. 

We first note that $\KO_2$ is $\KO_{\infty}-$absorbing since
$\KO_2$ is purely infinite and we may apply Kirchberg's theorem to conclude
that $\KO_2 \otimes \KO_{\infty} \simeq \KO_2.$ 

Thus by Thm. (B) of Szabo (\cite{S1}), 
there is only one Rokhlin
lift of the trivial $\KR-$action on $\ast = \Prim(\KO_2 \otimes \Cpct)$ to $\KO_2 \otimes \Cpct$ up
to cocycle conjugacy. 

Due to this, the above Rokhlin action must be cocycle conjugate to
the following Rokhlin action $\alpha_{\lambda}$ on $\KO_2 \otimes \Cpct$ 
(see Sec.\ (1) of Jacelon Ref. (\cite{J1})):

Consider a quasi-free Rokhlin $\KR-$action $\alpha_{\lambda}$ on $\KO_2.$ 
Let $s_1,s_2$ be generators of $\KO_2.$ Then, define

\begin{gather*}
\alpha_{\lambda}(s_1) = e^{ it }s_1 \\
\alpha_{\lambda}(s_2) = e^{ i \lambda t} s_2  
\end{gather*}
This action is Rokhlin if and only if $\lambda$ is irrational.
As described in Ref.\ \cite{J1}, the isomorphism class of the
crossed product depends on the sign of $\lambda.$

\begin{gather}
\CrPr{\KO_2}{\KR}{\alpha_{\lambda}} \simeq \KW \otimes \Cpct , \lambda \in \KR_{+}-\KQ \label{Eq1} \\ 
\CrPr{\KO_2}{\KR}{\alpha_{\lambda}} \simeq \KO_2 \otimes \Cpct , \lambda \in \KR_{-}-\KQ \label{Eq2}
\end{gather}
here $\KW$ is the Razak-Jacelon algebra \cite{J1}.

We take the crossed product of Eq.\ (\ref{Eq2}) above by the $\KR-$ action $\hat{\alpha}$ which is 
dual to the $\KR-$action $\alpha_{\lambda}.$
Since $\KO_2$ is unital, separable, nuclear, purely infinite and simple, and $\alpha_{\lambda}$ is Rokhlin,  
the induced dual-$\KR-$action $\hat{\alpha}$ on $\CrPr{\KO_2}{\KR}{\alpha_{\lambda}}$ is also Rokhlin (by 
Thm.\ (2.1), part (3) of Ref.\ \cite{BKR}). Thus, the action $\hat{\alpha}$ must be the unique 
(up to cocycle conjugacy) Rokhlin $\KR-$action on $\KO_2 \otimes \Cpct$ given by Szabo's Theorem above.

By Takai duality, the above crossed product is 
$\CrPr{(\KO_2 \otimes \Cpct)}{\KR}{\hat{\alpha}} \simeq \KO_2 \otimes \Cpct.$
Therefore, the action $\hat{\alpha}$ cannot give the `expected' T-dual $\KD \in \Bun_{S^1}(M_{2^{\infty}} \otimes \Cpct).$ 
Hence, there cannot be a Rokhlin $\KR-$action on $\KO_2 \otimes \Cpct$ which gives 
an element of $\Bun_{S^1}(M_{2^{\infty}} \otimes \Cpct)$ as a crossed product.
\end{proof}
The `reason' for the above seems to be that the T-dual $\KR-$action 
should have a $\KZ-$ stabilizer and the above action does not have any: Since
$\lambda$ is irrational, no value of $t$ will return the generators $s_i$ to themselves.

If $\lambda$ above was {\em rational}, the quasi-free action $\alpha_{\lambda}$ above would not be
Rokhlin. In particular, we would have by Thm.\ (3.2) of Dean (Ref.\ \cite{Dean}),
that in the case of rational $\lambda = p/q,$ the crossed product of $\KO_2$ by the quasi-free
$\KR-$action $\alpha_{p/q}$ would {\em always} be a mapping torus $C^{\ast}-$algebra of an $AF-$algebra 
$A(p,q).$

Since we are trying to T-dualize $M_{2^{\infty}}-$bundles over $S^1,$ the above facts would be consistent 
with Ex.\ (\ref{ExM2InfS1}), if the action on $\KO_2 \otimes \Cpct$
were cocycle conjugate to a tensor product of a non-Rokhlin quasi-free $\KR-$action on $\KO_2$ with the identity
$\KR-$action on $\Cpct.$ 

The above prescription should extend to $N-$torus bundles. We try to
apply the above prescription to the simplest such bundle.

\item{\underline{\bf{ {$M_{2^{\infty}}$-bundles over a space $X$ with a free circle action}}}}
\begin{itemize}
\item Let $X$ be an arbitrary supersymmetric Type II string theory background.
By Buscher's rules for fermionic T-duality (see Ref.\ \cite{MaldacenaB}),  the T-dual spacetime has the same $B-$field and metric,
but the Ramond-Ramond fluxes on the T-dual change. There is a circle action on the superspace fermionic coordinates  but
the underlying bosonic coordinates are unchanged by the circle action.
\item Hence, as argued in the previous section we would expect a duality between $C^{\ast}-$algebras with
$M_{2^{\infty}}-$fiber over $X$ and $C^{\ast}-$algebras with $\KO_2-$fiber over $X.$
\item It is difficult to analyze this problem with no further structure to $X,$ so we require that any $C^{\ast}-$algebra we associate to a space
in the problem be a $C_0(W)-$algebra where $W$ is a space over which $X$ fibers, that is we have a fibration map $q:X \to W.$

It is clear that in general this is not needed, since the circle action on the associated string theory is only in the fermionic directions so there should be no circle action on $X.$ This is different from the usual theory of $C^{\ast}-$algebraic T-duality (see Ref.\ \cite{JMRCBMS}), 
however requiring a fibration $q:X \to W$ helps us work out this example.

\item Let $\KD \in \Bun_{X}(M_{2^{\infty}} \otimes \Cpct).$ In addition to requiring $\KD$ and its dual to be a $C_0(W)-$algebra, we require
$X$ to be a principal circle bundle over $p:X \to W.$
By Buscher's rules for fermionic T-duality (see Ref.\ \cite{MaldacenaB}),  the 'T-dual' has the same topology and we will show that it can be chosen to have the
same $H-$flux.

\item We will see that with the method outlined in the previous section we obtain as a 'T-dual' of $\KD \in \Bun_{X}(M_{2^{\infty}} \otimes \Cpct)$ a $C^{\ast}-$algebra
$\KD' \in \Bun_{X}(\KO_2 \otimes \Cpct)$ which we will argue below should be associated to the fermionic T-dual space. 
\item If we were to calculate the
T-using the crossed product as normal, this crossed product would be difficult to calculate and also, it might not be possible to argue about its uniqueness. With the construction in the above section we get a unique T-duality diagram and hence a unique $C^{\ast}-$algebraic T-dual for $\KD.$

\begin{equation}
\begin{tikzcd}
&  & \C \arrow[dl, dashed] \arrow[dr,dashed] & \\
X \arrow[r, leftarrow, dashed, "\Prim"] \arrow[ddrr,"\pi"]  & \KD \arrow[dr, dashed] & & \KD^{\#} \arrow[dl, dashed] \\
&  & \B & \\ 
&  & W \arrow[u, leftarrow, dashed, "\Prim"] & 
\end{tikzcd}
\end{equation}

\item We associate the $C^{\ast}-$algebra 
$\B \simeq CT(W,\delta) \otimes_{C_0(W)} \KD'$
 where   $\KD' \in \Bun_{W}(M_{2^{\infty}} \otimes \Cpct)$ to $W,$ note that this is a $C_0(W)-$algebra.
\item We assume given a  on $\KZ-$action $\alpha \simeq \alpha_1 \otimes \alpha_2$ on $\B$ with 
$\alpha_1$ an automorphism of $CT(W, \delta)$ and $\alpha_2$ is an automorphism of $\KD'''$ which acts only on the
fiber of $\KD'''.$
\item The crossed product by the $\KZ-$action induced by $\alpha$ is 
$$\KD \simeq \CrPr{\B}{\KZ}{\alpha} \simeq (\CrPr{CT(W,\delta)}{\KZ}{\alpha_1}) \otimes_{C_0(W)} (\CrPr{\KD'}{\KZ}{\alpha_2}).$$
\item $\KD$ is a crossed product of $\B$ by the tensor product $\KZ-$action generated by $\alpha_1 \otimes \alpha_2, $ so $\Prim(\KD)$ is a fiber product of $X_1 \simeq \Prim(\CrPr{CT(W,\delta)}{\KZ}{\alpha_1})$ with 
$X_2 \simeq \Prim(\CrPr{\KD'}{\KZ}{\alpha_2})$ over $W.$

The crossed product $\CrPr{CT(W,\delta)}{\KZ}{\alpha_1}$ is isomorphic to $CT(X_1, \eta)$ for some $X_1$ with $\pi_1:X_1 \to W$ a principal circle bundle by the theory of continuous-trace algebras.

Now, $\alpha_2$ acts only on the fiber of $\KD'$, hence, by Ref.\ \cite{EEK1}, Remark (4.4),
the crossed product $\CrPr{\KD'}{\KZ}{\alpha_2}$ is a locally trivial $C^{\ast}-$bundle with spectrum $W \times S^1$ and fiber $\KO_2 \otimes \Cpct.$
By Ref.\ \cite{DP1}, there can be only one such bundle, the trivial one. Thus $\CrPr{\KD'}{\KZ}{\alpha_2}$ is isomorphic 
to $C_0(W \times S^1) \otimes \KO_2 \otimes \Cpct$ and hence its spectrum is $W \times S^1.$

Thus, the spectrum of $\KD$ is the fiber product 
$X_1 \times_W X_2 \simeq (X_1 \times_W (W \times S^1)).$

%
%
%
%

\item Suppose $\beta \simeq \beta_1 \otimes \beta_2$ was an automorphism of $\KD$ commuting with the action
$\hat{\alpha},$ 
where $\beta_1$ is an automorphism of $CT(X_1, \eta)$ above and $\beta_2$ a fiber-preserving automorphism
of $C_0(W \times S^1) \otimes \KO_2 \otimes \Cpct$ of the form $\beta_2 \simeq \id \otimes \beta_3$ where $\beta_3$ is a Rokhlin automorphism of $\KO_2 \otimes \Cpct.$

\item Note that $\KD$ is a stabilization of a unital $C^{\ast}-$algebra which is a section algebra of a $C^{\ast}-$bundle
with fiber $\KO_2$ and hence, this underlying bundle is $\KO_2-$absorbing.
By Cor.\ (2.22) of Ref.\ \cite{G2}, the set of Rokhlin circle actions are a dense, $G_{\delta}-$set in the set of
all continuous circle actions on $\KD.$  Thus we would expect many Rokhlin circle actions on this bundle. 
Any action $\alpha$ of this type induces an action of the form $\alpha \otimes \Ad(\lambda)$ on the stabilization
of this $C^{\ast}-$algebra (see proof of item (1) of Lemma\ (\ref{LemTorsionDCirc}) above, it is clear that the proof will carry over to this situation. Also,
note that the symbol $\B$ in that proof is different from the $C^{\ast}-$algebra $\B$ used here).
Thus there will be many Rokhlin circle actions on $\KD.$

%
\item Now consider, $\C \simeq \CrPr{\KD}{\KZ}{\beta}:$ By an argument similar to the above,
we find 
\begin{equation}
\Prim(\C) \simeq \Prim(\CrPr{CT(X_1,\eta)}{\KZ}{\beta_1}) \times_{W} 
\Prim( (C_0(W \times S^1)) \otimes \CrPr{\KO_2 \otimes \Cpct}{\KZ}{\beta_3})
\end{equation}

Now $\Prim(\CrPr{CT(X_1,\eta)}{\KZ}{\beta1})$ is a principal circle bundle over $X_1$ and 
$$\Prim((C_0(W \times S^1)) \otimes \CrPr{\KO_2 \otimes \Cpct}{\KZ}{\beta_3})$$ 
is by Ref.\ \cite{Katsura} 

\begin{equation}
\Prim(C_0(W \times S^1) \otimes ( \KO_2 \mbox{ fibration over} \otimes C_0(X) \otimes C(S^1)))
\end{equation}

(Here we have used Sec.\ (7.2) of Ref.\ \cite{Katsura} for the crossed product of $\KO_2 \otimes \Cpct$ by $\beta_3$.)
Hence, $\Prim(\C)$ is a nontrivial circle bundle over $X$ as expected.

\item Since $\beta$ commutes with the circle action $\hat{\alpha}$ it induces a circle action on $\C$. Quotienting
by this circle action gives the dual $C^{\ast}-$algebra.
\end{itemize}

\item{{\bf{\underline{$M_{2^{\infty}}-$bundles over the torus:}}}
\begin{itemize}
\item Let $p:\KT^N \to {\ast}$ as a trivial $N-$torus bundle. 
\item Consider a locally trivial $M_{2^{\infty}} \otimes \Cpct-$bundle ($\KD$) over a torus $\KT^N.$
\item As in the previous example, such bundles are classified by elements in
$H^{2n+1}(\KT^N, \KZ), n = 1,2,\ldots$   
\item Similar to the previous example we suppose given a $\KZ^N-$action $\alpha$ 
on $\KD$ which fixes $\Prim(\KD)$ and {\bf commutes with translation action on 
$\KD$ which is the natural lift of the torus action on $\KT^N.$}
\item Since the $\KZ^N-$action fixes $\Prim(\KD),$ it acts on the fiber $M_{2^{\infty}}$ 
of $\KD$ and also commutes with the natural translation action on $\KT^N.$
\item By a theorem of N.C. Phillips Ref.\ \cite{Phil}- see Thm.\ (5.15) and item (4) on pg. 1 
and the sections from Notation (5.7) to Lemma (5.10) on pg. 29 of Ref.\ \cite{Phil} as well-
$\KD$ is $M_{2^{\infty}}-$absorbing so Rokhlin actions of $\KZ^N$ are a dense
$G_\delta-$set in set of all actions of $\KZ^N$ on $\KD.$
Thus it should be possible to choose the $\KZ^N-$action above to be a Rokhlin $\KZ^N-$action.
 


\item In the absence of a characterization of the crossed product of a UHF algebra 
by $\KZ^N$ in the literature the correspondence space can't be determined. From the literature, it should be a
bundle of $\JZ-$stable, $UHF-$stable $C^{\ast}-$algebras over a point. 
\item There is no parametrization of cocycle conjugacy classes of $\KZ^N-$actions on $C^{\ast}-$algebras
like $\KD$ above. Hence it is not possible to determine the structure of the correspondence space
or the Topological T-dual.
\end{itemize}

%
%
%
}
\end{enumerate}

\section{Applications to T-duality in Type II String Theory \label{SecPhys}}
In this section we apply the formalism outlined above to study tree-level dualities in flux backgrounds in String Theory, i.e.,
spacetime backgrounds with sourceless Ramond-Ramond flux and $H-$flux. 
We begin by discussing the physics of such flux backgrounds in Subsec.\ (\ref{SubSecFluxCStar}). 
Then, in Subsec.\ (\ref{SubSecNCGFluxB}) we examine the noncommutative geometry of points in string backgrounds described
by $C^{\ast}-$algebras $\KD$ above and make a prediction about $D-$Brane charge. 
We apply the above formalism to tree-level dualities
in String Theory in Subsec.\ (\ref{SubSecThreeEx}).
\subsection{Flux Backgrounds and $C^{\ast}-$algebras \label{SubSecFluxCStar}}

We had described the $C^{\ast}-$algebraic formalism of Topological T-duality in Sec.\ (\ref{SecIntro})
above. We now use the example at the end of the previous Section to study certain phenomena in String Theory
which we claim are described by the crossed-product of section algebras of locally trivial bundles of strongly self-absorbing $C^{\ast}-$algebras by $\KR-$actions.

In Sec.\ \ref{SecIntro} we had mentioned that Topological T-duality describes spacetimes $X$ which are backgrounds for Type II
String Theory with a sourceless $H-$flux $H$ together with a circle action.  It is not clear whether the usual description of $D-$brane charges on string backgrounds $(X,[H])$ by twisted $K-$theory will work if other fluxes, for example Ramond-Ramond fields,
are turned on. 

It is interesting to ask if the $C^{\ast}$-algebra bundles studied in this paper (i.e., locally trivial bundles of 
self-absorbing $C^{\ast}$ algebras over a base space $X$) can be used to describe these string backgrounds.

We claim that the $C^{\ast}-$algebras
$\D$ in $\Bun_X(\A \otimes \Cpct)$ describe backgrounds of Type II string
theory with a circle action with various fluxes switched on (these are also termed flux backgrounds and were discussed in
Sec.\ (\ref{SecIntro}) above) as described below:
\begin{enumerate}
\item{{\bf $\A \simeq \KC$:}} If $\A \simeq \KC,$ the background has a sourceless $H-$flux
present and there is an equivariant gerbe with connection present
on the spacetime background and the $H-$flux is the (characteristic class of) the
curvature of the gerbe connection in $H^3(X,\KZ).$
There is a natural lift of the $S^1-$action on $X$ to a $\KR-$action on the gerbe covering the circle action on $X$ 
thus giving an equivariant gerbe on $X.$

The equivariant gerbe then defines a $S^1-$equivariant principal $PU-$bundle on the spacetime background up to
isomorphism of equivariant $PU-$bundles. The continuous-trace $C^{\ast}-$algebra on the spacetime background
is the associated bundle to this equivariant $PU-$bundle and has fiber $\Cpct$ with $PU$ acting on $\Cpct$ by $\Ad$
(see Ref.\ \cite{JMRCBMS, Kapustin} and references therein).

Thus, if $\A = \KC,$ $\D$ corresponds to backgrounds which have a sourceless $H-$flux switched on.

\item{{\bf $\A \neq \KC$:}} Suppose first that $X$ is compact, then, we have
the following (heuristic) argument: Cornalba et al in Ref.\ \cite{CCS} 
(Secs. (5.6.2) and (6.2)) consider a $Dp-$brane with a worldvolume gauge field $F$ and 
argue using String Field Theory that turning on a constant $RR-$ and $B-$field backgrounds causes the $Dp-$brane
worldvolume gauge field to change due to the backreaction by the open strings. The
backreaction has the form $F + \delta F$ where $\delta F$ is
of the form $\delta F = (B - *C^{(8)}).$
This implies that the $B-$field and the $RR-$flux together act only to shift
the worldvolume gauge field strength. 
Note that in the above $B-$field is only {\em shifted} by the $RR-$flux.
 
We now claim that if the background has a sourceless $H-$flux together and a sourceless
Ramond-Ramond flux is turned on, the above implies that the gerbe on the background has been
{\em deformed} into a (possibly different) gerbe with connection form $(B-*C^{(8)})$. 
We identify this deformed gerbe
as one of the higher gerbes of Cor.\ (3.42) of Ref.\ \cite{Pennig1}.
(Recall that Ramond-Ramond fields can also be viewed as curvatures of a gerbe on spacetime
(see Sec.\ (5.1) of Ref.\ \cite{Ruffino})). Note that these higher gerbes are the gerbes 
corresponding to locally trivial $C^{\ast}-$bundles with fiber $\A \otimes \Cpct$
discussed above. 

The twist of this higher gerbe can be used (not uniquely) to obtain a principal
$\Aut(\A \otimes \Cpct)-$bundle over the spacetime background
(see Cor.\ (3.42) of Ref.\ \cite{Pennig1} and following note).
Thus the higher gerbe gives an element $\D$ in $\Bun_X(\A \otimes \Cpct)$ by Ref.\ \cite{Dadarlat}. 

Note that with the above argument if the classification scheme of Dadarlat and Pennig is used
when $RR-$flux is turned on, the gerbe on $X$ defined by the $B-$field can only change
its characteristic class however, it still remains a gerbe. The associated $C^{\ast}-$algebra on the other
hand, changes from a continuous-trace $C^{\ast}-$algebra
to a locally trivial $C^{\ast}-$bundle with fiber $\A \otimes \Cpct.$

It would be interesting to generalize the above to equivariant gerbes, but for this,
we would need an analogue of the theory of the equivariant Brauer group for
$\KR-$actions on section algebras of locally trivial bundles with fiber a strongly self-absorbing
$C^{\ast}-$algebra. In this paper, we have managed to find lifts for various fibers
of physical interest based on the current literature, however, the general case is unknown
at present.

Hence, the case $\A \neq \KC,$ for any spacetime background of 
Type II theory $X$ above the $C^{\ast}-$algebra $\D$ in $\Bun_X(\A \otimes \Cpct)$ 
corresponds to backgrounds which have other fluxes switched on apart from the $H-$flux 
(see item (\ref{H3Summand}) in the list below). 
\end{enumerate}

Then, the construction proposed in Sec.\ (\ref{SecLift}) above
should give a duality relation which generalizes
Topological T-duality for ordinary spacetime backgrounds with sourceless $H-$flux to
a duality for the above spacetime backgrounds with various sourceless $H-$ and $RR-$ fluxes turned on.
Under this duality, the $C^{\ast}-$algebra 
$\D^{\#}$ (see Sec. \ref{SecLift} above)
is a 'dual $C^{\ast}-$algebra' to $\D.$ Now, the dual space $X^{\#}$ 
should be recovered from  $\D^{\#}$ as $X^{\#} = \Prim(\D^{\#})$. Note that if
$\D^{\#}$ is a locally trivial bundle of self-absorbing $C^{\ast}-$algebras over $X^{\#}$ then the characteristic
classes of the fluxes on it can be recovered from the characteristic class of $\D.$ If, on the other hand, $\D^{\#}$
is not a locally trivial bundle of self-absorbing $C^{\ast}-$algebras, techniques from noncommutative geometry
might be helpful in determining the dual space $X^{\#}$ (see the last part of Sec.\ (\ref{SecLift}) above).

In Refs.\ \cite{Brodzki1,Brodzki2,Brodzki3}, the authors define noncommutative spacetimes as $C^{\ast}-$algebras
satisfying certain conditions. Noncommutative $D-$branes correspond to elements of the $K-$homology
of the spacetime and their charges lie in the operator $K-$theory of these $C^{\ast}-$algebras.
We note that the $C^{\ast}-$algebras $\KA$ satisfy the UCT (by Ref.\ \cite{TomsWin}) and so
do the $C^{\ast}-$algebras $\KD.$ Hence, by Cor.\ (4.5) for Ref.\ \cite{Brodzki3}, it is possible
to define a noncommutative analogue of the $D-$brane charge for these spacetimes.

Hence, it is natural to identify a 'generalized $D-$brane charge' of a $D-$brane in $\KD$ 
with the operator $K-$theory of $\D.$

In Ref.\ \cite{BouMatSph}, Sec. (6.1.2) Bouwknegt and Mathai calculate the twisted $K-$theory of
a space $P$ which is a principal $SU(2)-$fibration over a $4-$manifold $M$ twisted by an integral $7-$form $H \in H^7(P,\KZ)$  using the operator $K-$theory of a bundle of Cuntz algebras on $P.$  In this paper the authors wish to define a 'T-dual' principal $SU(2)-$bundle
using a generalization of Topological T-duality. Since the authors are studying $SU(2)-$bundles over a four-dimensional base space they do not consider crossed products by $\KR$ as this paper does. It is interesting to note that in Section (9) item (2) of the same paper, the authors reject the possibility that a crossed-product of the Cuntz algebra bundle above by a $SU(2)-$action would give them the T-dual of the spacetimes they are examining since it turns out there cannot be an  isomorphism in twisted $K-$theory for crossed-products by 
$SU(2)-$actions.

Also, in Ref.\cite{MMS} pg. 334, Section (1)  Macdonald, Mathai and Saratchandran argue that not all higher twisted
$K-$theories on a space are obtained by twisting the $K-$theory of a space by cohomology data. As an example of this, they study 
higher twisted $K-$theory of a space obtained from $\KO_{\infty}-$bundles mentioned in the previous paragraph.
For an important class of examples, they demonstrate that the twist of the higher twisted $K-$theory obtained 
naturally defines a class $H:X \to S^n$ in the cohomotopy set of $X.$ 
However, for finite, connected, torsion free spaces these theories do correspond to a twisting by a cohomology class
on the space.

In this paper, we argue that the choice of a bundle of self-absorbing $C^{\ast}-$algebras on spacetimes with a circle action
is a good model for spacetime backgrounds containing $D-$branes with sourceless $H-$flux and possibly $R-R-$flux.
In addition, we argue that the Topological T-dual proposed in Sec.\ \ref{SecLift}
gives a physically sensible 'generalized T-duality' (see below for more details). 
 


\subsection{Noncommutative Geometry and $K-$theory of Flux Backgrounds \label{SubSecNCGFluxB}}
Now we study $X$ together with a fixed algebra $\KD \in \Bun_{X}(\A \otimes \Cpct)$ 
describing a Type II string theory background with sourceless $H-$flux, and
possibly $RR-$flux, from the point of view Noncommutative Geometry and $K-$theory within
the above formalism.
\begin{enumerate}
\item {\flushleft{{\bf Effects on Points:}}}
Consider the inclusion of a point in $\{ \ast \}$ in $X.$ The pullback of
the algebra $\D$ along the inclusion map gives a trivial $C^{\ast}-$bundle 
over $\{ \ast \}$ with fiber $\A \otimes \Cpct.$
Thus, we may associate the $C^{\ast}$-algebra $\A \otimes \Cpct$ with a point
$\{ \ast \}$ in this formalism. We may view different types of self-absorbing $C^{\ast}-$algebras
$\A$ as defining various types of {\em noncommutative points} on the worldvolume gauge theory
of $D-$branes of the space $X$ with
$H-$flux or $RR-$flux turned on. {\em We would expect this to be visible in the
worldvolume theory of $D0-$branes on $X.$}

\item {\flushleft{{\bf Effects on Branes and Brane Charges:}}}
In Ref.\ \cite{Kapustin}, A. Kapustin has argued that in the case of 
$m$ coincident $D$-branes, the gauge theory on the $D$-brane 
worldvolume is described by a gauge theory on an Azumaya algebra
with fiber $M_m(\KC)$ whose characteristic class is the
restriction of the background $H$-flux to the $D$-brane worldvolume.

It is well known, (see Ref.\ \cite{RaeWill}) that the stabilization of
an Azumaya algebra is a continuous-trace $C^{\ast}$-algebra with
{\em torsion} characteristic class and every continuous-trace 
$C^{\ast}$-algebra with torsion characteristic class arises in this
way.

In a recent paper, Dadarlat and Pennig (see Ref.\ \cite{DP2}, Thm.\ (2.8))
have shown that for $\A$ a self-absorbing $C^{\ast}-$algebra
the elements of $\Bun_{X}(\A \otimes \Cpct)$
with torsion characteristic class may be obtained from
the stabilization of $C^{\ast}$-bundles over $X$ with fiber $M_n(\A)$ for some $n.$
In Sec. (\ref{TorsTD}) above, we have calculated the Topological T-dual of 
such $C^{\ast}-$algebras.

This similarity with the case of continuous-trace $C^{\ast}-$algebras
is extremely interesting. The change from $M_n(\KC)$ to $M_n(\A)$ in the above
may be attributed to {\em  'noncommutative points'} discussed above.

\item {\flushleft{{\bf $K-$theory of $X:$}}} 
In String Theory, it is well-known that when the $H-$flux is zero
the $K$-theory of $X$ reflects properties of the $D$-brane
charge on the space (see Ref.\ \cite{Witten} for example).

It is well-known that when $\A = \KC,$ $K_{\ast}(\D )$ is the twisted
$K$-theory of $X$ with the twisting the class of $\D$ in 
$\Bun_{X}(\KC \otimes \Cpct) \simeq H^3(X,\KZ).$ 
In more detail: When the twist vanishes, $\D = C(X) \otimes \Cpct$ and $K_{\ast}(\D)$ is
the ordinary $K$-theory of $X.$ 
When the twist does not vanish, but $\A = \KC,$ $K_{\ast}(\D)$ is the
twisted $K$-theory of $X,$ with the twist the class determined by 
$[\D]$ in $H^3(X,\KZ).$
It is conjectured that $D$-brane charges in
String Theory with background $H$-flux take values in twisted $K$-theory of $X.$
(see Ref.\ \cite{JMRCBMS} and references therein).

If we assume that the formalism outlined at the beginning of the current
section is true, the $D-$brane charge on $X$ is given by the $K-$theory of $\D$
where $\D$ is in $\Bun_{X}(\A \otimes \Cpct)$ for any strongly 
self-absorbing $C^{\ast}$-algebra $\A.$

{\underline{First consider the untwisted case:}} 
We consider $\D = C_0(X,\A) \simeq C_0(X)\otimes \A.$
By Ref.\ \cite{DadarlatK} Sec.\ (4), before Sec.\ (4.1), 
for each strongly self-absorbing $C^{\ast}$-algebra
$\A,$ the assignment $X \mapsto K_{\ast}(C(X) \otimes \A) \simeq K_{\ast}(C(X,\A))$
is a multiplicative generalized cohomology theory on finite $CW$-complexes. By 
Ref.\ \cite{Pennig1} $K_0(C(X,\A))$ is the Grothendieck group of
isomorphism classes of (finitely generated and projective)
Hilbert-$\A-$module bundles over $M.$

In Ref.\  \cite{Pennig}, the author argues that this cohomology theory 
can be represented by a commutative symmetric spectrum $KU^{\A}_{\bullet}$ (see
Ref.\ \cite{Pennig} Sec.\ (1)). In particular, for $\A =\JZ,$ the Jiang-Su algebra, the
resulting theory is topological $K-$theory, while for $\A$ an infinite UHF-algebra
of type $p^{\infty}$, the resulting theory is a localization of $K-$theory at the
prime $p$ with spectrum $KU[1/p].$
\newline
It would be interesting to see
if one could find $D-$brane configurations whose charge group 
is this group for some choice of $\A$ for a Type II background with the
correct sourceless fluxes turned on.
\newline
All self-absorbing $C^{\ast}-$algebras
$\A$ may be written as direct limits of matrix algebras. It is
interesting to speculate whether this fact could be used to 
build such configurations, perhaps as argued in the next example:

{\flushleft {\bf Example (CAR-Bundles):}} Consider a trivial $M_{2^{\infty}}$-bundle over $X.$
We can view
$$C_0(X,M_{2^{\infty}}) \simeq \varinjlim C(X,M_{2^k}(\KC))$$
where each inclusion corresponds to a doubling map (see Ref.\ \cite{JMRCBMS} and references
therein).
Each $C_0(X,M_{2^{k}}(\KC))$ factor can be viewed as the Azumaya bundle of the worldvolume gauge theory of
a stack of $2^k$  $D-$branes.
We view the injection $C_0(X, M_2(\KC)) \into C_0(X,M_4(\KC))$ by doubling
a matrix as associated to the worldvolume theory of four $D-$branes bound into a 'dimer' of two $D-$branes each.
We view the entire direct limit as the worldvolume theory associated to a stack of such dimers. While this cannot happen if background
$RR-$flux is not present (since the worldvolume gauge theory would then be given by the Azumaya 
algebra $C_0(X,M_4(\KC))$), it would be interesting to see if 
this or a similar construction would give a volume gauge theory described by 
$C_0(X, M_{2^{\infty}})$ when sourceless background $RR-$flux was turned on.

\underline{Now consider the twisted case:}
For the twisted case, we need to consider section algebras of nontrivial $C^{\ast}-$bundle with
fibers stabilized self-absorbing $C^{\ast}-$algebras.
In Ref.\ \cite{Pennig}, the author demonstrates that the operator algebraic $K-$theory of
these $C^{\ast}-$algebras  is a twisted version of the above, a higher twisted $K-$theory.

Thus, for a space with Ramond-Ramond flux concentrated in odd degrees together with
$H-$flux, we could take $\A=\JZ$ or $\A = M_{2^{\infty}}$ in the above.
Hence, {\em the charge group of $D-$branes on a space with Ramond-Ramond flux concentrated in odd
degrees should be a higher twisted topological $K-$theory or a twisting
of the localization of topological $K-$theory at the prime $2.$ } It would be interesting if this could be calculated
from String Theory.

It would be interesting to speculate which string theory backgrounds would have excitations which
have these cohomology theories as charges. Ref.\ \cite{Uranga} demonstrates that D-branes
in a certain flux background geometry have charges in the group $\KZ_p:$ Stacks of D-branes in these
backgrounds annihilate in sets of $p$ D-branes at a time. It might be possible to
construct backgrounds which have D-brane charges in the $K-$groups of the above $C^{\ast}-$algebras.

\end{enumerate}

\subsection{Physical Examples and Tree-Level Dualities}
\label{SubSecThreeEx}
We now consider some important examples of the above formalism and the physical transformations associated with them:
First note that when all background fluxes are turned off apart from the $H-$flux, the formalism above reduces to the
$C^{\ast}-$algebraic T-duality of Mathai and Rosenberg (\cite{MRCMP}).
However, if some of the background fluxes are turned on, we would obtain the fermionic
T-duality of Berkovits and Maldacena (\cite{MaldacenaB}) and the timelike T-duality of Hull (\cite{HullPope}). 

In the above, the $C^{\ast}-$algebraic T-duality of Mathai and Rosenberg models the T-duality symmetry of Type II String theory which is true symmetry of Type II string theory valid to all orders in perturbation theory. However, the Fermionic T-duality of Berkovits and Maldacena and the Timelike T-duality of Hull are currently only tree-level symmetries of Type II string theory, loop corrections to string theory might destroy this symmetry-thus while
the spacetime might have such symmetries on a coarse scale, if examined on fine scales, there need be no such symmetry.
We nonetheless consider these tree-level symmetries here, since we are only interested in the coarse features of the spacetime background (in particular, its algebraic topology).

We argue that these three examples correspond the crossed-product construction above for three different choices
for $\A$-namely $\A = \KC, M_{2^{\infty}}$ or $\JZ$ respectively together with a specific choice of the geometry
of $X.$

(As was argued above, we associate the $C^{\ast}-$algebras $\KC, M_{2^{\infty}}, \JZ$
with spacetime backgrounds with $H-$flux only, with Ramond-Ramond flux only and with both $H-$flux and Ramond-Ramond 
flux respectively.)

\begin{enumerate}
\item {\underline{\bf T-duality with $H-$flux:}} 
As was explained in detail at the beginning of this paper, the case $\A = \KC$ corresponds
to continuous-trace bundles over $X$ and these describe a space with
background sourceless $H-$flux as 
described by Mathai and Rosenberg in Ref.\ \cite{MRCMP}. The non-stable
case (corresponding to $M_n(\KC)-$bundles over $X$)
was examined earlier by Kapustin in Ref.\ \cite{Kapustin}.

By Thm.\ (\ref{ThmDPrim}) above, 
for any $\D \in \Bun_X(\A \otimes \Cpct),$
$\Prim(\D) \simeq X.$ For the case $\A = \KC,$ (see Ref.\ \cite{RaeRos,MRCMP}),
we have the stronger result $\hat{\A} \simeq X.$ 

\item {\underline {\bf Fermionic T-duality:}}
\begin{enumerate}
\item {\underline {\em Mathematical Example:}}
Now we consider the case $\A = M_{2^{\infty}}$ (also called the $CAR-$algebra). 
The sections of locally trivial bundles with fiber $M_{2^{\infty}} \otimes \Cpct$ are
operator-valued fields on $X$ which satisfy
the Canonical Anticommutation Relations.
\newline
The $K-$theory of $M_{2^{\infty}}$ is $K_0(M_{2^{\infty}}) \simeq \KZ[1/2]$ and $K_1(M_{2^{\infty}}) \simeq 0.$
The positive cone of $M_{2^{\infty}}$ is
$$K_0(M_{2^{\infty}})_{+} \simeq \{ a/2^{k} |a \in \KZ, a \geq 0, k \geq 0 \}.$$ 
As a result $K_0(M_{2^\infty})^{\times}_{+} \simeq \KZ$ since it is
a cyclic group on one generator namely $2 \in \KZ[1/2].$ Hence, $\Bun_{S^1}(M_{2^{\infty}} \otimes \Cpct) \simeq \KZ$ and there can
be nontrivial fiberings of $M_{2^{\infty}}$ over $S^1.$


In Sec.\ (\ref{CStarEx}) Item (1), the Topological T-dual-in the sense of Sec.\ (\ref{SecLift})-of $C^{\ast}-$algebras in
$\Bun_{S^1}(M_{2^{\infty}} \otimes \Cpct)$ was calculated. It was shown that the Topological T-dual must be isomorphic
$\KO_2 \otimes \Cpct$ and by the argument in Sec.\ (\ref{SecLift}) this must be the unique Topological T-dual as a noncommutative
space.

When $X = \{ \ensuremath{\mbox{pt}} \},$ a space with only one point,  it had been remarked in that example that the crossed product
should be viewed as a noncommutative space ($\KO_2$)
which is a `noncommutative fibration' with circle fibers
over a simple algebra (the CAR algebra again, see Sec. (\ref{CStarEx}) Item (1) above)-whose $\Prim$ must be a point.
We also showed in that example that no {\em Rokhlin} $\KR-$action can give this space as a crossed product.

We consider a supersymmetric background $\KR^{(8,1)} \times S^1.$ 
We place a graviphoton flux on the background and perform a fermionic
T-duality. The T-dual should be non-anti-commutative superspace
over a circle. We identify the original space before fermionic T-duality with an element of
$\Bun_{S^1}(M_{2^{\infty}} \otimes \Cpct).$ 

We claim that the fermionic T-dual of a supersymmetric theory over a circle is non-anti-commutative superspace
over a point-which we identify with the unique Topological T-dual $\KO_2 \otimes \Cpct$ above.
The T-dual circle action would then correspond to the rotation of the generators of non-anti-commutative superspace by a phase.
This should be matched with the Buscher's rules for fermionic T-duality: It can be seen that the physical T-dual is a circle
with the fermionic directions `twisted' by the T-dual graviphoton flux into NAC superspace.

In Sec.\ (\ref{CStarEx}) Item (2), the Topological T-dual-in the sense of Sec.\ (\ref{SecLift})-of specific $C^{\ast}-$algebras in
$\Bun_{X}(M_{2^{\infty}} \otimes \Cpct)$ was calculated for $X$ a space with a free circle action.
We can see that the calculation is a `fibering' of the previous calculation in item (1) there.

Due to all the above, we strongly suggest that the Topological T-dual in the sense of Sec.\ (\ref{SecLift}) of a $M_{2^{\infty}}-$bundle over a space $X$ with a free circle or torus action is a non-anticommutative superspace over $W \simeq X/S^1.$ In our proof we find that
the topological type of the T-dual bundle will vary depending on the details of the lifted circle action, but we suggest that this be
fixed by matching the result to Buscher's rules.

In Sec.\ (\ref{CStarEx}) Item (3), we argued that the Topological T-dual in the sense of Sec.\ (\ref{SecLift}) may be extended
to trivial $\KT^N-$bundles over a point. We also proved that in this case there is a unique Topological T-dual in this extended sense,
but pointed out the difficulty of calculating it. For example, it is clear that the Topological T-dual $C^{\ast}-$algebra is a
$\JZ-$stable, $UHF-$stable $C^{\ast}-$bundle over a point, but its not clear what the Topological T-dual {\em is.}

\item {\underline {\em Relation to Physics:}}
We speculate that the above transformation is the {\em  `fermionic T-duality'} of Maldacena-Berkovits.
In Ref.\ \cite{MaldacenaB}, the authors define a new type of T-duality transformation, at the tree level, which T-dualizes in
a fermionic direction. For the duality to be nontrivial, the analogue of the $B-$field in the fermionic
direction (the 'graviphoton flux') has to be nontrivial. 

From the paper of Maldacena-Berkovits, the effect of a fermionic T-duality is to leave the topology of 
the space alone. The analogue of the Buscher rules for fermionic T-duality are quoted below in the bispinor formalism
of Ref.\ \cite{MaldacenaB} from Sec. (2.4) of the same paper (see Ref.\ \cite{MaldacenaB} for notation and details):

\begin{gather}
-\frac{i}{4} e^{\phi'} F^{'\alpha \hat{\beta}} = -\frac{i}{4} e^{\phi} F^{\alpha \hat{\beta}} - 
\epsilon^{\alpha} {\hat{\epsilon}}^{\hat{\beta}} C^{-1} \\
\phi^{'} = \phi + \frac{1}{2} \log C
\end{gather}
where as usual unprimed symbols denote the original fields while primed ($'$) symbols denote the dual fields. Also
$\phi$ is the dilaton, $C$ is the $\theta = \hat{\theta} = 0$ component of the Kalb-Ramond field component $B_{11},$ 
identified with the graviphoton flux
$(\epsilon^{\alpha},\hat{\epsilon}^{\hat{\alpha}})$
are spinors and
$F^{\alpha \hat{\beta}}, F^{'\alpha \hat{\alpha}}$ are bispinors constructed from the
Ramond-Ramond Fluxes.

It can be seen from the above that the T-duality transformation does not change 
the metric or $B-$fields but changes the Ramond-Ramond flux and dilaton on the background. 

Fermionic T-duality was later studied for a two dimensional supertorus with $B-$field and graviphoton flux in
Ref.\ \cite{TwoDSuperT2}. In this paper, the T-dual is the result of an ordinary 'bosonic' T-duality followed by
a 'fermionic T-duality'. This is because the isometry being dualized has components in the fermionic direction.

Consider the example in Sec.\ (\ref{CStarEx}) Item (2) above. We associate 
the $C^{\ast}-$algebra  $\KD \in \Bun_{S^1}(M_{2^{\infty}} \otimes \Cpct)$
to a superspace with body a point with one fermionic direction with periodization in the fermionic direction. 
If there is Ramond-Ramond flux present on the original spacetime,
it gives a class in $H^1(S^1,\KZ)$ which is
part of the characteristic class of the $M_{2^{\infty}}-$bundle above.

\item {\underline {\em $K-$theory of $\KD$ and $D-$brane charge:}}
The Topological T-dual $C^{\ast}-$algebra in the sense of Sec.\ (\ref{SecLift})
to $\KD$ in Part (a) above is -by the Example in Part (a) above-$\KO_2 \otimes \Cpct.$

Now, suppose $X \neq \ensuremath{ \mbox{ pt } }.$ Since $K_0(\KO_2 \otimes \Cpct)$ is zero, $K_0(\KO_2 \otimes \Cpct)^{\times}_{+}$
 is zero as well, hence the elements of the spectral sequence tableaux of Ref. \cite{Dadarlat} after Thm. (4.2) are identically zero.  The dual bundle is trivial with zero 
characteristic class since the spectral sequence computing $E_{\KO_2}^{\ast}(X)$ is identically zero (see Ref.\ \cite{Dadarlat} after Thm. (4.2)).
Also, note that since $R^{\times}_{+} = K_0(\A)^{\times}_{+}$ is zero and the characteristic class of the Ramond-Ramond flux can only lie in 
$H^1(X,R^{\times}_{+}) \simeq 0$ there can be no Ramond-Ramond flux on the T-dual. This is
consistent with the Buscher's rules for fermionic T-duality above as well-by 
Ref.\ \cite{MaldacenaB} the T-dual has imaginary Ramond-Ramond flux under fermionic T-duality.

We tentatively identify the geometry of the T-dual  as a point times 'non-anti-commutative superspace'  (see Ref.\ \cite{TwoDSuperT2}).
It is interesting to note that the generators of the Cuntz algebra (which is the T-dual)
$$
S_1 S_1^{\ast} + S_2 S_2^{\ast} = 1,     S_j^{\ast} S_j = 1
$$
 look remarkably similar to the
defining relations for non-anti-commutative variables $\theta^a$
$$
\theta^a \theta^b + \theta^b \theta^a = 1.
$$
We tentatively associate the Cuntz algebra to the ring of functions on  'non-anti-commutative superspace '- the T-dual of superspace with graviphoton
flux turned on.

Consider $\KD$ as the trivial element of $\Bun_{S^1}(M_{2^{\infty}} \otimes \Cpct).$ 
The $K-$theory of the space is, by Kunneth's theorem $K_0(\KD) \simeq \KZ[1/2], K_1(\KD) \simeq 0.$

\item {\underline {\em $K-$theory of $X$ and twisting:}} We had remarked above that the space $X \simeq \Prim(\KD)$ should be viewed as having
'noncommutative points' given by the $C^{\ast}-$algebra $\A \simeq M_{2^{\infty}}.$ 
We note that this may be viewed as changing the
$K-$theory of the space $X:$ To calculate the $K-$theory of the above space we should actually calculate
$K_{\ast}(\KD)$ (and not $K_{\ast}(C_0(X))$ as usual). 

In the above example, 
we have a trivial $(M_{2^{\infty}} \otimes \Cpct)-$bundle
over $X$ and hence, the $K-$theory of the spacetime background should be given by 
$K_{\ast}(C_0(X, M_{2^{\infty}}))$ and not $K_{\ast}(C_0(X,\KC))$ as is usual. 

 Now, we have that $C_0(X,M_{2^{\infty}}) \simeq C_0(X) \otimes M_{2^{\infty}}$
hence by the Kunneth theorem for $C^{\ast}-$algebras (see Blackadar's book Ref.\ \cite{BlacBook}, V.1.5.10, pg. 417), we have that 
$$K_{\ast}(C_0(X,M_{2^{\infty}})) \simeq K_{\ast}(C_0(X)) \otimes K_{\ast}(M_{2^{\infty}})$$
(the $\Tor^{\KZ}_{1}$ term is zero since $K_{\ast}(M_{2^{\infty}})$ is torsion free). 
This gives us the $K_0$ group of $X$ as $K_0(C_0(X)) \otimes \KZ[1/2]$ and
$K_1$ group of $X$ as $K_1(C_0(X)) \otimes \KZ[1/2]$. We suspect that the above groups store information about the fermionic
part of the supersymmetric background, but it would be interesting to examine this further.

If we restrict ourselves to $M_{2^{\infty}}-$bundles over $X,$ we cannot get the 'higher twists' of Ref.\ \cite{Pennig1}, to
see these we need to study other $C^{\ast}-$algebras $\A.$

\end{enumerate}

{\flushleft{\bf Remark:}} It would be interesting to test the above idea on a better example, perhaps a graded 
$C^{\ast}-$algebra $\KG$ which was a Cuntz algebra bundle over a space with noncommutative torus fibers. 
We would also like to make an analogy with Ref.\ \cite{TwoDSuperT2} and  naturally associate this $C^{\ast}-$algebra to 
the the T-dual of a supertorus with $B-$field and graviphoton flux.

One way to calculate the above $C^{\ast}-$algebra $\KG$ might be as follows: Following the example of Ref.\ \cite{MRCMP}, 
Sec. (5) we could consider $C^{\ast}(\KH_{\KZ}) \otimes \KO_2 \otimes \Cpct.$ This is
a trivial noncommutative $\KO_2 \otimes \Cpct-$bundle over the noncommutative space associated to
 $C^{\ast}(\KH_{\KZ})\otimes \Cpct$. In Ref.\ \cite{MRCMP}, Sec. (5), 
it was shown that $C^{\ast}(\KH_{\KZ}) \otimes \Cpct$ was a noncommutative two-torus fibration over a circle. 
It was also argued in Section (5) of Ref.\ \cite{MRCMP} that $C^{\ast}(\KH_{\KZ}) \otimes \Cpct$ was an example of a 'noncommutative T-dual' of a three-torus with $H-$flux. 

We would like to construct a $C^{\ast}-$algebra $\KG$ with $\Prim(\KG) \simeq \KT ^3$  with $H-$flux and $RR-$flux
which would have $C^{\ast}(\KH_{\KZ}) \otimes \KO_2$ as a crossed product by a $\KR^2-$action
lifting the natural circle action on $\KT^3.$ It would be interesting to generalize
the construction on Section (5) of Ref. \cite{MRCMP} to this situation.

However, the above lifting result won't work for lifting $\KR^2-$actions, so at present its not clear how to do the above 
calculation. By analogy with Ref.\ \cite{TwoDSuperT2}, this noncommutative space $\KG$  should
possess a rich set fermionic and bosonic T-duals which would appear as Morita Equivalences of the $C^{\ast}-$algebra $\KG.$

{\flushleft{\bf Remark:}}Generalizing the above machinery, we conjecture that if the spacetime background possesses $N=k$ supersymmetry we should associate
to it $M_{2^{\infty}} \otimes \underset{k \mbox{ times }}{\ldots} \otimes M_{2^{\infty}}$ bundles over $X$. Note that such bundles are isomorphic to $M_{2^{\infty}}-$bundles
over the space $X$ since 
$$
M_{2^{\infty}} \otimes \underset{k \mbox{ times }}{\ldots} \otimes M_{2^{\infty}} \simeq M_{2^{\infty}}
$$
by definition of the $CAR$ algebra as a direct limit of an infinite tensor product of $M_2(\KC)$ (see Ref.\ \cite{JMRCBMS} and references therein for details).
Thus, we conjecture that we should {\bf associate one $CAR
-$algebra factor} in $\A \otimes \Cpct$ {\bf with each supersymmetric direction}.

\item {\underline{\bf{Timelike T-duality:}}} \label{H3Summand}
\begin{enumerate}
\item {\flushleft \underline{\bf Mathematical Example:}}
From Cor.\ (4.6) of Ref.\ \cite{Dadarlat}, if $\A = \JZ$ the Jiang-Su algebra,
$\Bun_X(\JZ \otimes \Cpct) \simeq \bigoplus_{k \geq 1} H^{2k+1}(X,\KZ).$ 
Thus, the characteristic class of locally trivial bundles over $X$ 
with fiber $\JZ \otimes \Cpct,$ the stabilized Jiang-Su algebra
might encode the characteristic
classes of background $RR$-fluxes.

Further, by Cor.\ (4.7) of Ref.\ \cite{Dadarlat},
there is a natural map $\KC \to \JZ$ which implies that there is
a natural transformation of cohomology theories 
$T:E^{\ast}_{\KC}(X) \to E_{\JZ}(X).$ By the same Corollary,
$E^{1}_{\KC}(X) \simeq H^3(X,\KZ)$ is a natural direct summand
of $E^1_{\JZ}(X) \simeq \Bun_{X}(\JZ \otimes \Cpct).$ 
Hence locally trivial bundles over $X$ with
fiber the stabilized Jiang-Su algebra $\JZ \otimes \Cpct$ might describe
a space with background $H$-flux 
together with background $RR$-fluxes. 

\item {\flushleft \underline{\em{Relation to Physics:}}}
We speculate that the above transformation is the 'Timelike T-duality' transformation
of Hull (see Ref. \cite{HullPope} for details).

Consider spacetime backgrounds of the form $X \times \KR$ where the $\KR-$direction
is the time variable and the $X$ direction has no circle action. 
To the spacelike part  $X$ of $X \times \KR$ with sourceless, constant $RR-$flux and $H-$flux
we associate $\KD \in \Bun_{X}(\JZ \otimes \Cpct)$. 
We mimic periodization in the timelike direction by considering
$\Ind^{\KR}_{\KZ}(\KD, \alpha)-$ these
are stabilized Jiang-Su bundles over the circle.
$X \times S^1 \simeq  \Prim(\Ind^{\KR}_{\KZ}(\KD, \alpha))$ where $\alpha$ is
an automorphism of $\KD$ reflecting the properties of the periodization above.

It is not possible to calculate the Topological T-dual in Sec.\ (\ref{SecLift}) for $\KZ \otimes \Cpct,$
not enough is known about Rokhlin group actions on the Jiang-Su algebra. In this case, even though
the result won't cover Rokhlin actions, we attempt to identify the crossed product 
$\CrPr{\Ind^{\KR}_{\KZ}(\KD, \alpha)}{\KR}{\alpha}$ with the timelike T-dual.
This is reasonable as this is only a preliminary inquiry.

If we take $\alpha$ to be independent of $x \in X,$ then,
$$
\Ind^{\KR}_{\KZ}(\KD \otimes \Cpct,\alpha) \simeq C_0(X) \otimes 
\Ind^{\KR}_{\KZ}(\JZ\otimes \Cpct, \alpha).
$$
We calculate this
fully in the next part.

\item {\flushleft \underline{\em {$K-$theory and $D-$brane charge:}}}
By Thm. (2.9) of Ref.\ \cite{JiangSu} we find the $K-$theory of $\JZ$ as
$$
(K_0(\JZ), K_0(\JZ)_{+}, K_1(\JZ), [1_{\JZ}]) \simeq (\JZ, \JZ_{+} \cup \{ 0 \}, \{ 0 \}, 1)
$$
Now, by the discussion around Eqs.\ (\ref{SSACl1}, \ref{SSACl2}) above, and
by Cor.\ (4.6) of Ref.\ \cite{Dadarlat}, 
$\Bun_{S^1}(\JZ \otimes \Cpct) \simeq K_0(\JZ)^{\ast}_{+} \simeq \{ 1 \}$
as an abelian group. Hence there can't be any nontrivial Jiang-Su bundles over
the circle. Thus, $\KD \simeq C(S^1) \otimes \JZ \otimes \Cpct.$

The $K-$theory of $K_{\ast}(C_0(X) \otimes \KT_{\alpha}(\JZ \otimes \Cpct))$
now follows from the previous part. Using the Kunneth
theorem and the fact that the $K-$theory of the Jiang-Su algebra has no torsion
(see Ref.\ \cite{TomsWin} and references therein for details)
the $K-$theory groups are given by 
$$
K_{\ast}(C_0(X) \otimes C(S^1) \otimes \JZ \otimes \Cpct) \simeq K_{\ast}(C_0(X \times S^1)) \otimes K_{\ast}(\JZ).
$$
By the previous part, $K_0(\JZ) \simeq \KZ$ and $K_1(\JZ) \simeq 0.$
Hence $K_{0}(C_0(X) \otimes \KT_{\alpha}(\JZ \otimes \Cpct) ) \simeq K^0(X \times S^1)$
and $K_{1}(C_0(X) \otimes \KT_{\alpha}(\JZ \otimes \Cpct)) \simeq K^1(X \times S^1).$
Thus, the dual has the same $K-$theory as the topological $K-$theory of $X \times S^1.$
This is to be expected since the dual gains the timelike direction along which
timelike T-duality acted as a new spatial dimension (see Ref.\ \cite{HullPope}).

\item {\flushleft \underline{\em {$K-$theory and twisting:}}}
As in the previous example we should calculate $K_{\ast}(C_0(X) \otimes \JZ)$ 
since the natural $K-$theory of $X$ with $\JZ$ present is not given by
$K_{\ast}(C_0(X))$ but by $K_{\ast}(C_0(X) \otimes \JZ).$ 
By Ref.\ \cite{Pennig1}, this is isomorphic to $K^{\ast}(X).$

Physically, we could interpret this as the following: The twisting of the
$D-$brane charges by the $H-$flux seem to vanish when $RR-$flux is turned on.
It would be interesting to see if this can be determined from String Theory calculations.

%
%
\end{enumerate}

\end{enumerate} 

\section{Summary and Conclusions \label{SecConclusion}}
To summarize,we associate a spacetime background $X$ with $H-$flux and sourceless Ramond-Ramond flux to a section algebra ($\D$) of
a locally trivial fiber bundle with fiber a fixed self-absorbing $C^{\ast}-$algebra ($\A$).

This has the following consequences:
\begin{itemize}
\item The $C^{\ast}-$algebra of functions on a point changes from $\KC \otimes \Cpct$-in the presence of $H-$flux only- to 
$\A \otimes \Cpct$-where $\A$ depends on which background fields are present: In the presence of Ramond-Ramond flux only $\A = M_{2^{\infty}}$ and in the presence of $H-$flux and Ramond-Ramond flux $\A = \JZ$.

Thus, the presence of sourceless Ramond-Ramond flux causes the appearance of 'noncommutative points' of various types depending on whether other fluxes are present or not. It should be possible to detect this by examining the worldvolume theory of $D0-$branes or the correct matrix model. 

\item The stabilized $C^{\ast}$-algebra of functions on the brane worldvolume changes from being strongly Morita equivalent to $C_0(U)$ locally to being strongly Morita equivalent to $C_0(U, \A \otimes \Cpct)$ locally. The obstruction to the global strong Morita equivalence is 
as follows:

If only $H-$flux is present the $C^{\ast}-$algebra of functions on the brane worldvolume is only locally strongly 
Morita equivalent to $C_0(X).$ The obstruction to it being 
{\em globally } strongly Morita equivalent to $C_0(X)$ is the gerbe on the brane worldvolume (see 
Ref.\ \cite{Pennig1}) whose gerbe curvature is the $H-$flux. 

If Ramond-Ramond flux is present the stabilized $C^{\ast}-$algebra of functions on the brane worldvolume is  locally strongly Morita equivalent to $C_0(X,\A).$ The obstruction to it being {\em globally} strongly Morita equivalent to $C_0(X,\A)$ is a {\em higher gerbe } on the brane worldvolume (see Ref.\ \cite{Pennig1} Sec. (1)). 

The work in this paper implies that the 'twist' of this higher gerbe encodes all the sourceless background fields present ($H-$flux and Ramond-Ramond flux). 
It would be interesting to see if a connection-like structure on this higher gerbe gave fields with the correct transformation laws  in analogy with gerbe connections on gerbes on the worldvolume of $D-$branes in a background sourceless $H-$flux. 
\item As argued in Refs.\ \cite{Kapustin} the
worldvolume gauge theory for ordinary stack of  $N$ 
$D-$branes may be written using section algebras of locally trivial $M_N(\KC)-$bundles (Azumaya algebras) over the
$D-$brane worldvolume $X.$

Based on the arguments we would expect the worldvolume gauge theory of a possibly infinite stack of $D-$branes in background sourceless
$RR-$flux to be expressible in terms of the $C^{\ast}-$algebras like $\KD$ above which are elements of $\Bun_X(\A \otimes \Cpct).$ 
It would also be interesting to see if specific D-brane configurations correspond  to specific $C^{\ast}-$algebras
$\KD \in \Bun_X(\A \otimes \Cpct)$.

\item In addition we argue about the lift of $S^1-$actions on $X$ to unique Rokhlin $\KR-$actions on $\D$ up to cocycle conjugacy.
We have seen that the lift is possible for a certain class of fiber algebras and might not be possible for other fiber algebras.
We propose a procedure for finding the Topological T-dual which always gives a T-dual $C^{\ast}-$algebra for
any fiber algebra and which agrees with the crossed product $C^{\ast}-$algebra T-dual when there is a Rokhlin lift of the circle
action on $X$ present. We also characterize the T-dual $C^{\ast}-$algebra in Thm.\ (\ref{TDAll}) above.

We use this formalism to calculate the Topological T-dual $C^{\ast}-$algebra for some examples. A more detailed analysis of the structure of
the T-dual calculated in Thm.\ (\ref{TDAll}) might give some more interesting-and possibly physically relevant-examples by analogy with Topological T-duality for continuous-trace $C^{\ast}-$algebras.

We propose two physical examples of this extension of Topological T-duality  in Subsection (\ref{SubSecThreeEx}) above:Fermionic T-duality and Timelike T-duality.

\end{itemize}

%
%

\section{Acknowledgements}

I thank Professor Jonathan M. Rosenberg of the University of
Maryland, College Park, for all his advice and help.

I thank the Department of Mathematics, Harish-Chandra Research Institute,
Allahabad, for support in the form a postdoctoral fellowship
during which the first draft of this paper was written.

I thank the School of Mathematics, NISER, HBNI, Bhubaneshwar for support in
the form of a Visiting Assistant Professorship during the writing of part of this paper.

I thank the Mathematical and Physical Sciences Division, School of Arts and Sciences, Ahmedabad University, Ahmedabad for all their help and
support. 

\providecommand{\href}[2]{#2}

\address{Division of Physical and Mathematical Sciences,\\
School of Arts and Sciences, Ahmedabad University, Ahmedabad, India.\\
\email{ashwin.s.pande@gmail.com, ashwin.pande@ahduni.edu.in}}

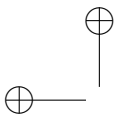
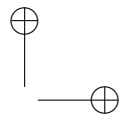
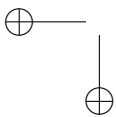
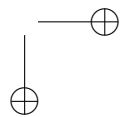
\end{document}